\newif\ifarxiv%
\arxivtrue%

\documentclass[acmsmall]{acmart}

\usepackage{amsmath}
\usepackage{dkenv}
\usepackage{mathpartir}
\usepackage{thm-restate}

\usepackage{tikz}
\usetikzlibrary{arrows,automata}

\usepackage{enumerate}
\usepackage{xspace}
\usepackage[capitalise]{cleveref}
\usepackage{caption}
\usepackage{subcaption}
\usepackage{graphicx}
\usepackage{booktabs}
\usepackage{array}
\usepackage{rotating}
\usepackage{xcolor}

\newcommand\flush{\mathop{\mathsf{flush}}}

\newcommand\X{\mathit{pkt}_\mathsf{X}}
\newcommand\Topo{\mathsf{Topo}}
\newcommand\addr{\alpha}
\newcommand\Red{{\color{red}\mathsf{R}}}
\newcommand\Blue{{\color{blue}\mathsf{B}}}

\newcommand\Tangerine{{\color{orange}\mathsf{T}}}
\newcommand\pcap{\textsc{pcap}}
\newcommand\Purple{{\color{violet}\mathsf{P}}}

\newcommand\push{\mathop{\mathsf{push}}}
\newcommand\pop{\mathop{\mathsf{pop}}}
\newcommand\pifopush{\mathop{\textsc{push}}}
\newcommand\pifopop{\mathop{\textsc{pop}}}

\newcommand\hide[1]{}

\newcommand\len[1]{|#1|}

\newcommand\set[2]{\ensuremath{\{#1 \mid #2\}}}
\newcommand\Imp{\Rightarrow}
\newcommand\subs{\mathrel{\subseteq}}

\renewcommand\phi\varphi%

\newcommand\vertlabel[1]{\begin{turn}{90}#1\end{turn}}
\newcommand\pic[1][tree]{\includegraphics[width=.4\textwidth]{gallery/#1.png}}
\newcolumntype{C}{>{\centering\arraybackslash}m{.4\textwidth}}

\usepackage{xifthen}
\newcommand{\ir}[3][]{%
  \ifthenelse{\isempty{#1}}{%
    \ifthenelse{\isempty{#2}}{%
        \inferrule*[]{{ }}{#3}%
    }{%
        \inferrule*[]{#2}{#3}%
  }}{%
    \inferrule*[right={\sc #1}]{#2}{#3}%
  }%
}

\theoremstyle{definition}
\newtheorem{remark}{Remark}

\usepackage{bookmark}

\usepackage{todonotes}

\renewcommand{\vec}[1]{{#1}s}



\setcopyright{rightsretained}
\acmPrice{}
\acmDOI{10.1145/3622845}
\acmYear{2023}
\copyrightyear{2023}
\acmSubmissionID{oopslab23main-p376-p}
\acmJournal{PACMPL}
\acmVolume{7}
\acmNumber{OOPSLA2}
\acmArticle{269}
\acmMonth{10}
\received{2023-04-14}
\received[accepted]{2023-08-27}

\begin{document}

\title{Formal Abstractions for Packet Scheduling}
\ifarxiv%
\subtitle{Extended Version}
\settopmatter{printacmref=false}
\fi%

\author{Anshuman Mohan}
\orcid{0000-0002-6803-9767}
\email{amohan@cs.cornell.edu}
\affiliation{%
  \institution{Cornell University}
  \country{USA}
}

\author{Yunhe Liu}
\orcid{0000-0003-4677-8902}
\email{yunheliu@cs.cornell.edu}
\affiliation{%
  \institution{Cornell University}
  \country{USA}
}

\author{Nate Foster}
\orcid{0000-0002-6557-684X}
\email{jnfoster@cs.cornell.edu}
\affiliation{%
  \institution{Cornell University}
  \country{USA}
}

\author{Tobias Kappé}
\orcid{0000-0002-6068-880X}
\email{t.kappe@uva.nl}
\affiliation{%
  \institution{Open University}
  \country{the Netherlands}
}
\affiliation{%
  \institution{ILLC, University of Amsterdam}
  \country{the Netherlands}
}

\author{Dexter Kozen}
\orcid{0000-0002-8007-4725}
\email{kozen@cs.cornell.edu}
\affiliation{%
  \institution{Cornell University}
  \country{USA}
}

\begin{abstract}
Early programming models for software-defined networking (SDN) focused on basic features for controlling network-wide forwarding paths, but more recent work has considered richer features, such as packet scheduling and queueing, that affect performance.
In particular, \emph{PIFO trees}, proposed by Sivaraman et al., offer a flexible and efficient primitive for \emph{programmable} packet scheduling.
Prior work has shown that PIFO trees can express a wide range of practical algorithms including strict priority, weighted fair queueing, and hierarchical schemes.
However, the semantic properties of PIFO trees are not well understood.

This paper studies PIFO trees from a programming language perspective.
We formalize the syntax and semantics of PIFO trees in an operational model that decouples the scheduling policy running on a tree from the topology of the tree.
Building on this formalization, we develop compilation algorithms that allow the behavior of a PIFO tree written against one topology to be realized using a tree with a different topology.
Such a compiler could be used to optimize an implementation of PIFO trees, or realize a logical PIFO tree on a target with a fixed topology baked into the hardware.
To support experimentation, we develop a software simulator for PIFO trees, and we present case studies illustrating its behavior on standard and custom algorithms.
\end{abstract}

\maketitle

\section{Introduction}

Over the past decade, programmable networks have gone from a dream to reality~\cite{deep-programmability}.
But why do network owners want to program the network?
Although there has been some buzz around trendy topics like self-driving networks~\cite{bingzhe20} and in-network computing~\cite{netcache,netchain,netpaxos}, network owners often have less flashy priorities: they simply want to build networks that provide reliable service under uncertain operating conditions.
Doing this well in practice requires not just fine-grained control over routing, but also prioritizing certain packets over others---i.e., controlling packet scheduling---using algorithms that operate at line rate.

Today, most routers support just a few scheduling algorithms---e.g., strict priority, weighted fair queueing, etc.---that are baked into the hardware.
An administrator can select from these algorithms and configure their parameters to some extent, but they typically cannot implement entirely new algorithms.
To get around this, \citet{Sivaraman16} proposed a new model for \emph{programmable} packet schedulers based on a novel data structure called a \emph{PIFO tree}.\footnote{A PIFO is just a priority queue (\emph{push-in-first-out}) that is additionally defined to break ties in \emph{first-in-first-out} (\emph{FIFO}) order.}
This relatively simple data structure can be instantiated to not only realize a wide range of well-studied packet scheduling algorithms, but also compose them hierarchically.
It seems likely that PIFO trees will be supported on network devices in the near future---indeed, the original paper on PIFO trees presented a detailed hardware design and demonstrated its feasibility, and researchers have also started to explore how the PIFO abstraction can be emulated on fixed hardware~\cite{Alcoz20, vass2022}.

Informally, a PIFO tree associates a PIFO with each of its nodes.
The leaf PIFOs hold buffered packets and the internal PIFOs hold scheduling data.
To insert a packet into a PIFO tree, we first determine the leaf where the packet should be buffered, and then walk down the tree from the root to that leaf.
At each internal node along this path, we insert into the node's PIFO a downward reference to the next node.
At the leaf node we insert the packet itself, and this completes the insertion.
The subtlety comes from choosing the \emph{rank}s with which downward references and packets are inserted into PIFOs; these ranks are determined by a program called a \emph{scheduling transaction}.
Conversely, releasing a packet employs a simple, fixed algorithm: pop the root's PIFO to get an index to its most favorably ranked child, and recurse on that child until arriving at a leaf.
Popping the leaf gives us the tree's most favorably ranked packet, which is then emitted.

Despite their operational simplicity, relatively little is known about how the topology of PIFO trees affects the expressivity of their scheduling algorithms.
We are also not aware of any work that studies how to translate scheduling algorithms written for a tree with one topology onto another tree with a different topology.
From a theoretical perspective, this question is interesting because it offers insights into the expressiveness of PIFO trees.
From a practical perspective, it is important because it allows chip designers to focus on developing high-speed implementations with fixed topologies, freeing up hardware resources for implementing the scheduling logic.

In this paper we develop the first formal account of PIFO trees, using tools and techniques from the field of programming languages to model their semantics.
We then embark on a comprehensive study of PIFO trees, examining how expressiveness is affected by variations in topology.
This leads to our main result: an algorithm that can compile a scheduling policy written against one PIFO tree onto another PIFO tree of a different topology.
We furthermore develop a reference implementation of our algorithm, as well as a simulator that we use to validate the algorithm.
Finally, we compare the behavior of the simulator against scheduling algorithms running on a state-of-the-art switch.

More concretely, this paper makes the following contributions:
\begin{itemize}
    \item
    We present the first formal syntax and semantics for PIFO trees, emphasizing the separation of concerns between the topology of a tree and the scheduling algorithm running on it.
    \item
    We study the semantics of PIFO trees in terms of the permutations they produce, and use this as a tool to formally distinguish the expressiveness of PIFO trees based on their topologies.
    \item
    We propose a fast algorithm that compiles a PIFO tree into a regular-branching PIFO tree, and a slower algorithm that compiles a PIFO tree into an arbitrary-branching PIFO tree.
    Our algorithms are accompanied by formal proofs characterizing when compilation is possible.
    \item
    We provide an implementation of the first algorithm, as well as a simulator for PIFO trees.
    We use this simulator to validate that our compiler preserves PIFO tree behavior, and compare PIFO trees against scheduling algorithms implemented on a hardware switch.
\end{itemize}

\noindent
Overall, this work takes a first step toward higher-level abstractions for specifying scheduling algorithms by developing compilation tools, and suggests directions for future work on scheduling, in the networking domain and beyond.

\section{Overview}%
\label{sec:overview}

We start by discussing what is perhaps the most obvious way to implement software-defined scheduling: a programmable priority queue.
In this model, incoming packets are ranked using a (user-defined) \emph{scheduling transaction}, and a packet is enqueued according to its rank.
In general, the scheduling transaction may read and write state.
To dequeue a packet, we simply remove the most favorably ranked packet (i.e., the one with the with lowest rank or, equivalently, highest priority).

Although such an approach is simple and relatively easy to implement, it is unable to express a basic feature of many packet scheduling algorithms: the ability to reorder buffered packets after they are enqueued.
To address this shortcoming, we introduce PIFO trees~\cite{Sivaraman16}, which extend the simple programmable queue model and support such reordering.

\subsection{Programmable Priority Queues}%
\label{sec:pifos}

Suppose you are responsible for programming a switch where incoming packets can be divided into two flows, coming to you from data centers in $\Red$ochester and $\Blue$uffalo respectively.
The flows are to be given equal priority up to the availability of packets.
This last caveat is important: if one flow becomes inactive and the other stays active, then the active flow is given free rein to use the excess bandwidth until the other flow starts up again, at which point it again receives its fair share.

Fortunately, the queue in your network device is programmable: it allows you to specify a \emph{scheduling transaction}, which assigns an integer rank to each packet and may also update some internal state.
The packet is then inserted into a priority queue, or PIFO, which orders the packets by rank.
Whenever the link becomes available, a different component in the network device removes the most favorably ranked packet from the queue and transmits it.

The scenario where we aim for an equal split between $\Red$ and $\Blue$ traffic can be implemented by assigning ranks to incoming packets such that the contents of the queue interleave these flows whenever possible, and by maintaining some state.
Specifically, our aim will be to maintain a queue that has one of the following three forms at any given time.
Note that, in keeping with network queueing tradition, the most favorably ranked item in the queue is on the \emph{right}, which we call the \emph{head} of the queue.
These queues are read from right to left, in the order they will be dequeued.
\begin{mathpar}
    \Blue_n, \ldots, \Blue_1,~~(\Red, \Blue)^*
    \and
    \Red_n, \ldots, \Red_1,~~(\Red, \Blue)^*
    \and
    (\Red, \Blue)^*
\end{mathpar}
Here, we write $(\Red, \Blue)^*$ to stand in for a balanced (possibly zero) number of interleaved $\Red$ and $\Blue$ packets interleaved in either order, such as $\Red, \Blue, \Red, \Blue$ or $\Blue, \Red, \Blue, \Red, \Blue, \Red$.

Let us see how these three cases accept a new $\Red$ or $\Blue$ packet.
We just have a PIFO, so our power is limited: all we can do is assign the incoming packet a rank and then push the packet into the PIFO at that rank.
The packet is inserted among the previously buffered packets without affecting the relative order of those packets.
We distinguish based on the three possible shapes of the queue:
\begin{itemize}
    \item
    To enqueue a $\Blue$ into $\Blue_n, \ldots, \Blue_1,~~(\Red, \Blue)^*$, give it any rank that puts it within the list of unbalanced $\Blue$s.
    Making it $\Blue_{n+1}$ puts it on the far left and preserves arrival order within the~$\Blue$ flow.
    To enqueue an $\Red$, put it just before or just after $\Blue_1$ (depending on whether the item at the head of the queue is $\Red$ or $\Blue$).
    This extends the initial (balanced) part of the queue.
    \item
    To enqueue a $\Blue$ or an $\Red$ into $\Red_n, \ldots, \Red_1,~~(\Red, \Blue)^*$, act symmetrically to the previous case.
    \item
    To enqueue a $\Blue$ or an $\Red$ into $(\Red, \Blue)^*$, assign it a rank that puts it on the very left.
    This creates a queue of either the first or the second kind.
\end{itemize}

\subsection{Achieving Balance within Flows}%
\label{sec:motivating_trees}

The solution proposed for sharing bandwidth between $\Red$ and $\Blue$ works well thus far, but now a complication arises.
As it turns out, traffic from $\Red$ochester can be further divided by destination: some packets are bound for $\Tangerine$oronto and some for $\Purple$ittsburgh.
Traffic to these sites also needs to be balanced equally in a~1:1 split, but \emph{within} the share of traffic allocated to $\Red$ packets.
That is, every two packets released should contain one $\Red$ and one $\Blue$, and every two~$\Red$ packets released should contain one $\Tangerine$ and one $\Purple$.
Again, this is up to the availability of packets.

We now highlight a particular case that may arise when trying to program a scheduling transaction that adheres to these constraints, and show that the desired split cannot be achieved using our single programmable queue.
Suppose $\Tangerine$ is initially silent, while $\Purple$ and $\Blue$ are active.
In this case, our scheduler should balance $\Purple$ and $\Blue$ \emph{evenly}: in the absence of $\Tangerine$ traffic, $\Purple$ traffic gets to take up the entire share allocated to $\Red$ traffic.
Under this policy, the queue might take the following form:
\[\Blue_3, \Blue_2, \Purple_2, \Blue_1, \Purple_1\]
So far, so good: opportunities for transmission are divided evenly between $\Purple$ and $\Blue$, up to availability of packets.
But now a $\Tangerine$oronto-bound packet $\Tangerine_1$ arrives, and our scheduling transaction needs to insert it somewhere in this queue.
With an eye to the 1:1 split \emph{within} the $\Red$ flow, $\Tangerine_1$ can go either before or after $\Purple_1$ but it really does need to go before $\Purple_2$---after all, if this is not the case, then opportunities for $\Red$ traffic are not being split evenly between $\Tangerine$ and $\Purple$.
This leaves us with three possible choices for inserting $\Tangerine_1$, which would result in one of the following queues:
\begin{mathpar}
    \Blue_3, \Blue_2, \Purple_2, \Tangerine_1, \Blue_1, \Purple_1
    \and
    \Blue_3, \Blue_2, \Purple_2, \Blue_1, \Tangerine_1, \Purple_1
    \and
    \Blue_3, \Blue_2, \Purple_2, \Blue_1, \Purple_1, \Tangerine_1
\end{mathpar}
However, \emph{all} of these choices violate the 1:1 contract between $\Red$ and~$\Blue$---recall that both~$\Tangerine$ and~$\Purple$ packets are still $\Red$ packets, so the options above all propose to send two $\Red$s in a row while an unbalanced $\Blue$ languishes in the back.
Queues that \emph{do} satisfy the balance requirements include
\begin{mathpar}
    \Blue_3, \Purple_2, \Blue_2, {\Tangerine_1}, \Blue_1, \Purple_1
    \and
    \Blue_3, \Purple_2, \Blue_2, \Purple_1, \Blue_1, {\Tangerine_1}
\end{mathpar}
but note that these require the reordering of previously buffered packets relative to each other.
Since our programmable queue model allows us to assign a rank only to the incoming packet, \emph{it is impossible to achieve the desired behavior using a single PIFO}.
Moreover, it was not obvious at the onset that the algorithm would require the reordering of previously buffered packets.

\subsection{Introducing PIFO Trees}

\citet{Sivaraman16} propose a remarkably elegant solution to the problem discussed above: instead of a single PIFO, use a \emph{tree} of PIFOs.
The leaves of this tree correspond to flows, and each leaf carries, in its PIFO, packets ranked according to rules local to the flow.
Meanwhile, the internal nodes of the tree contain PIFOs that prioritize the opportunities for transmission \emph{between classes of traffic}, without referring to a particular packet within those classes.
This model provides a layer of indirection, which lets us eke out more expressiveness from the simple PIFO primitive while maintaining the flexibility of being able to program the relative priorities of each flow and subflow.

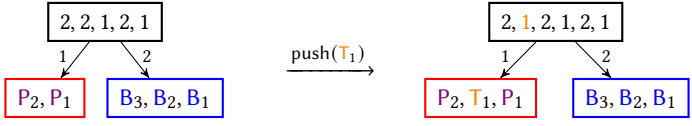
\begin{figure}
\begin{center}
\begin{tikzpicture}[->, >=stealth', auto, bend angle=10]
\tikzstyle{dot}=[state, inner sep=1.3pt, minimum size=0pt, draw=black, fill=black]
\tikzstyle{box}=[inner sep=0pt, minimum height=0.5cm, draw=black, thick]
\small
  \node[box] (eps) {\phantom{b}$2, 2, 1, 2, 1$\phantom{b}};
  \node[box, draw=red, thick] (N0) [below of=eps, xshift=-8mm] {$\phantom{b}\Purple_2, \Purple_1\phantom{b}$};
  \node[box, draw=blue, thick] (N1) [below of=eps, xshift=8mm] {$\phantom{b}\Blue_3, \Blue_2, \Blue_1\phantom{b}$};
  \path (eps) edge node[swap, pos=.65, xshift=2pt, yshift=-2pt] {$\scriptstyle 1$} (N0);
  \path (eps) edge node[pos=.65, xshift=-2pt, yshift=-2pt] {$\scriptstyle 2$} (N1);
  \node[box] (epsA) [right of=eps, node distance=60mm] {$\phantom{b}2, {\color{orange}1}, 2, 1, 2, 1\phantom{b}$};
  \node[box, draw=red, thick] (N0A) [below of=epsA, xshift=-10mm] {$\phantom{b}\Purple_2, \Tangerine_1, \Purple_1\phantom{b}$};
  \node[box, draw=blue, thick] (N1A) [below of=epsA, xshift=10mm] {$\phantom{b}\Blue_3, \Blue_2, \Blue_1\phantom{b}$};
  \path (epsA) edge node[swap, pos=.65, xshift=2pt, yshift=-2pt] {$\scriptstyle 1$} (N0A);
  \path (epsA) edge node[pos=.65, xshift=-2pt, yshift=-2pt] {$\scriptstyle 2$} (N1A);
  \node [right of=eps, node distance=30mm, yshift=-5mm] {$\xrightarrow{\push (\Tangerine_1)}$};
\end{tikzpicture}
\end{center}
\caption{A PIFO tree allows buffered packets to be reordered.}%
\label{fig:pifo_leaf_push}
\end{figure}

To illustrate how a PIFO tree operates, we now show how this model can be used to tackle the problem posed by the $\Red$/$\Blue$/$\Tangerine$/$\Purple$ traffic scenario discussed prior.
The PIFO tree that we will use is depicted in \Cref{fig:pifo_leaf_push} on the left.
Every node in this tree carries a PIFO, which we render, as before, with the most favorably ranked item on the right.
An internal node carries transmission opportunities for its children (here, indexed $1$ and $2$), while leaf nodes carry packets in their PIFOs.
The tree on the left holds our original five packets, before the arrival of $\Tangerine_1$. We maintain the left leaf as the $\Red$ leaf and the right leaf as the $\Blue$ leaf; we show this with colored borders.

Dequeueing the tree on the left triggers two steps.
First, we dequeue an element from the PIFO in the root; in this case we obtain $1$, a reference to the left child.
Second, we look at the left leaf, where dequeueing the PIFO returns the packet $\Purple_1$.
We emit this packet.
In general, the $\pop$ operation recurses into a tree until it reaches a leaf, following the path guided by the references popped from the internal nodes.
Repeating $\pop$ until the tree on the left is empty would yield the sequence $\Blue_3, \Blue_2, \Purple_2, \Blue_1, \Purple_1$, meaning that the tree on the left matches the setup of our problematic example.

Now $\Tangerine_1$ arrives. To $\push$ $\Tangerine_1$ into the tree, we perform two steps:
\begin{enumerate}
    \item
    Since $\Tangerine_1$ belongs to the $\Red$ flow, we insert it into the PIFO of the left leaf.
    We must assign~$\Tangerine_1$ a rank; as noted before, $\Tangerine_1$ needs to be ranked more favorably than $\Purple_2$, but either position relative to $\Purple_1$ is allowed.
    We place~$\Tangerine_1$ after~$\Purple_1$ in an attempt to stay closer to arrival order.
    \item
    We enqueue a \textcolor{orange}{$1$} index (shown in orange because it was enqueued on account of a $\Tangerine$ packet) into the root node.
    We get to pick {\color{orange}$1$}'s rank before inserting it into the root PIFO;\@ we choose a rank that maintains 1:1 harmony between the $1$s and $2$s.
\end{enumerate}
The resulting tree is on the right in \Cref{fig:pifo_leaf_push}.
Note how repeatedly popping elements from this tree gives the desired ordering $\Blue_3, \Purple_2, \Blue_2, \Tangerine_1, \Blue_1, \Purple_1$.
Critically, when the {\color{orange}$1$} at the root is popped, it causes the release of~$\Purple_2$ and not~$\Tangerine_1$.
This is because the index $1$ refers to an \emph{opportunity} for the left child to propose a packet, not to any individual packet enqueued there.
Indeed, the packet $\Tangerine_1$ has already been released by the time {\color{orange}$1$} is dequeued.
The release of $\Tangerine_1$ was triggered by an older $1$ index.

As before, we \emph{just} used PIFOs---at enqueue, each item was assigned a rank and inserted without affecting the relative order of the older items in the PIFO---but the on-the-fly target switching that we have just seen, in which an index enqueued on account of one packet may eventually lead to the release of a different packet, is elegant and powerful.
This is how PIFO trees facilitate the relative reordering of buffered packets despite the limitations of their relatively simplistic primitives.

\subsection{Implementations and Expressiveness}

A practical implementation of PIFO trees needs to operate at line rate---state-of-the-art devices achieve tens of terabits per second, which corresponds to 10 billion operations per second~\cite{tofino}.
\citet{Sivaraman16} proposed a hardware implementation, which overlays the PIFO tree on a mesh of interconnected hardware PIFOs, and a subset of this mesh is used to pass packets and metadata between parents and children on each operation.
This is a very flexible approach, capable of accommodating more or less\footnote{The number of different scheduling policies to control each internal node is also limited; refer to op.\ cit.\ for further details.} any PIFO tree whose number of internal nodes and leaves does not exceed the number of hardware PIFOs.
However, a fully connected mesh comes at a cost: it essentially requires the hardware to be configurable in terms of the connections that are used, which induces overhead in terms of performance, chip surface, cost, and complexity.

One could imagine a refinement of this hardware model, where the topology of the tree is fixed.
This raises the question: \emph{how much would a fixed topology of PIFO trees limit packet scheduling?}
In this motivating section we have shown that a single PIFO (i.e., a PIFO tree consisting of just one leaf) is less expressive than a two-level PIFO tree, but it is not immediately clear how this distinction generalizes.
Conversely, \emph{is it possible for one PIFO tree to express the behavior of another, even if their topologies differ?}
A constructive proof of such a correspondence would yield a compilation procedure, which could then be exploited to implement a user-designed PIFO tree on fixed hardware.

\subsection{Technical Contributions}

Our aim is to undertake a comprehensive study of PIFO trees as a semantic model for programmable packet queues.
Our technical contributions, and the remainder of this paper, are organized as follows.
\begin{itemize}
    \item[\S\ref{sec:synsem}.]
    We provide a rigorous model of the PIFO tree data structure, including its contents, operations, and well-formedness conditions.
    \item[\S\ref{SEC:LIMITS}.]
    Using this model, we formally prove that PIFO trees with $n$ leaves are less expressive than PIFO trees with $m$ leaves when $n < m$.
    By extension, this means that if one looks at the class of PIFO trees of degree $k$, taller PIFO trees are strictly more expressive.
    \item[\S\ref{SEC:EMBEDDING}.]
    We develop the notion of \emph{(homomorphic) embedding} as a tool for showing that the scheduling behavior of one PIFO tree can be replicated by another.
    \item[\S\ref{sec:embedding_algs}.]
    We propose two procedures to compute embeddings between PIFO trees, if they exist.
    These algorithms map a packet scheduling architect's scheduling policy, written against whatever tree the architect finds intuitive, onto the fixed PIFO tree actually implemented in hardware.
    \item[\S\ref{sec:implementation}.]
    We implement a simulator for PIFO tree behavior, as well as the algorithm from \S\ref{sec:embedding_algs}.
    We use this simulator to compare PIFO trees against their embeddings, as well as against standard algorithms implemented on a state-of-the-art programmable hardware switch.
\end{itemize}

\section{Structure and Semantics}%
\label{sec:synsem}

We now give a more formal definition of the syntax and semantics of PIFO trees.
Let's fix some notation.
When $K$ is a set and $n \in \mathbb{N}$, we write $K^n$ for the set of $n$-element lists over $K$.
We denote such lists in the plural (e.g., $\vec{k} \in K^n$), write $|\vec{k}|$ for the list length $n$, read the $i$-th element (for $1 \leq i \leq n$) with $\vec{k}[i]$, and write $\vec{k}[k/i]$ for the list $\vec{k}$ with the $i$-th element replaced by~$k$ $\in K$.

\subsection{Structure}
Since the topologies of PIFO trees are important for our results, we isolate them into a separate type, which we will momentarily use as a parameter.
Conceptually, these are very straightforward: a tree topology is just a finite tree that does not hold any data.
\begin{definition}
The set of \emph{(tree) topologies}, denoted $\Topo$, is the smallest set satisfying the rules:
\begin{mathpar}
    \inferrule{~}{%
        * \in \Topo
    }
    \and
    \inferrule{%
        n \in \mathbb{N} \\
        \vec{t} \in \Topo^n
    }{%
        \mathsf{Node}(\vec{t}) \in \Topo
    }
\end{mathpar}
Here, $*$ is the topology of a single-node tree.
Given a list $\vec{t}$ of $n$ topologies, $\mathsf{Node}(\vec{t})$ is the topology of a node with $\vec{t}$ as its children; the topology of the $i$-th child of $\mathsf{Node}(\vec{t})$ is given by $\vec{t}[i]$.
\end{definition}

\begin{example}
The topology of the PIFO tree from \Cref{fig:pifo_leaf_push} is given by $\mathsf{Node}(ts)$, where $ts$ is a two-element list of topologies, with both entries equal to $*$.
\end{example}

As mentioned, PIFO trees are trees where each node holds a PIFO\@.
A \emph{leaf} node uses its PIFO to hold packets, while an \emph{internal} node carries (1) a list of its children, and (2) a PIFO that holds valid references to those children.
We can now formally define the set of PIFO trees of a given topology.

\begin{definition}
We fix an opaque set $\mathsf{Pkt}$ of \emph{packets}, and a totally ordered set $\mathsf{Rk}$ of \emph{ranks}.
For any set $S$, we also presuppose a set $\mathsf{PIFO}(S)$ of PIFOs holding values from $S$ and ordered by $\mathsf{Rk}$.

The set of \emph{PIFO trees} of a topology $t \in \Topo$, denoted $\mathsf{PIFOTree}(t)$, is defined inductively by
\begin{mathpar}
    \inferrule{%
        p \in \mathsf{PIFO}(\mathsf{Pkt})
    }{%
        \mathsf{Leaf}(p) \in \mathsf{PIFOTree}(*)
    }
    \and
    \inferrule{%
        n \in \mathbb{N} \\
        \vec{t} \in \Topo^n \\
        p \in \mathsf{PIFO}(\{1, \dots, n\}) \\\\
        \forall 1 \leq i \leq n.\ \vec{q}[i] \in \mathsf{PIFOTree}(\vec{t}[i])
    }{%
        \mathsf{Internal}(\vec{q}, p) \in \mathsf{PIFOTree}(\mathsf{Node}(\vec{t}))
    }
\end{mathpar}
Here, $\mathsf{Leaf}(p)$ is a PIFO tree consisting of just one leaf holding the packet-PIFO $p$, and $\mathsf{Internal}(\vec{q}, p)$ is a PIFO tree whose root node holds $p$, an index-PIFO, as well as $\vec{q}$, a list of PIFO tree children.
\end{definition}

\subsection{Semantics}
We can now define the partial function $\pop$, which takes a PIFO tree and outputs a packet and an updated PIFO tree.
We assume that PIFOs themselves support a partial function $\pifopop: \mathsf{PIFO}(S) \rightharpoonup S \times \mathsf{PIFO}(S)$, which returns the most favorably ranked element of the PIFO and an updated PIFO\@.

\begin{definition}
For all topologies $t \in \Topo$, define $\pop: \mathsf{PIFOTree}(t) \rightharpoonup \mathsf{Pkt} \times \mathsf{PIFOTree}(t)$ by
\begin{mathpar}
    \inferrule{%
        \pifopop (p) = (\mathit{pkt}, p')
    }{%
        \pop (\mathsf{Leaf}(p)) = (\mathit{pkt}, \mathsf{Leaf}(p'))
    }
    \and
    \inferrule{%
        \pifopop (p) = (i, p') \and
        \pop (\vec{q}[i]) = (\mathit{pkt}, q')
    }{%
        \pop (\mathsf{Internal}(\vec{q}, p)) = (\mathit{pkt}, \mathsf{Internal}(\vec{q}[q'/i], p'))
    }
\end{mathpar}
In rules of this form, here and hereafter, the types of the variables are inferred from context.
For example,~$p$ is a packet-PIFO in the rule on the left but an index-PIFO in the rule on the right.
\end{definition}

When run on a leaf, $\pop$ simply applies $\pifopop$ to its PIFO, returning the released packet and the updated node.
When $\pop$ is run on an internal node, it applies $\pifopop$ to its PIFO;\@ the returned value is an index $i$ pointing to a child node, on which $\pop$ is called recursively.
The packet returned by the $i$-th child is also the present node's answer, along with a PIFO tree reflecting the effects of $\pop$ping.

It is worth pointing out that $\pop$ is not a total function: it may be the case that $\pop(q)$ is undefined, most obviously when when $q$ is a leaf or internal node with an empty PIFO (hence the call to $\pifopop$ is undefined).
This can also cascade up the tree, e.g., when an internal node's PIFO points to its $i$-th child, but the recursive $\pop$ call fails there.
We expand on how to prevent this in \Cref{sec:wf}.

\smallskip
To define the operation $\push$, which inserts a packet into a PIFO tree, we will diverge slightly from \citet{Sivaraman16}.
There, a PIFO tree chooses a leaf to enqueue the packet; the algorithm then walks from this leaf to the root and enqueues a downward reference at each node, where the rank of each enqueue is computed by a scheduling transaction attached to each non-root node.
In our treatment, we assume that $\push$ is supplied with a \emph{path} that contains all of this information precomputed; the path is constructed by an external \emph{control}, which we will define in a moment.

\begin{definition}
The set of paths for a topology $t \in \Topo$, denoted $\mathsf{Path}(t)$, is defined as follows.
\begin{mathpar}
    \inferrule{%
        r \in \mathsf{Rk}
    }{%
        r \in \mathsf{Path}(*)
    }
    \and
    \inferrule{%
        \vec{t} \in \Topo^n \and
        1 \leq i \leq n \and
        r \in \mathsf{Rk} \and
        pt \in \mathsf{Path}(\vec{t}[i])
    }{%
        (i, r) :: \mathit{pt} \in \mathsf{Path}(\mathsf{Node}(\vec{t}))
    }
\end{mathpar}
\end{definition}

Intuitively, a path is a list $(i_1, r_1) :: \dots :: (i_n, r_n) :: r_{n+1}$, where the first $n$ elements $(i_j, r_j)$ contain the index $i_j$ of the next child and the rank $r_j$ at which this index should be enqueued, and the last element $r_{n+1}$ is the rank with which the packet should be enqueued in the leaf PIFO\@.
For example, the path $(1, 10) :: (3, 5) :: 6$ means that the scheduling transaction wants us to perform three steps:
\begin{enumerate}
    \item
    Enqueue the index $1$ in the root node with rank $10$.
    \item
    Enqueue the index $3$ in the first child of the root node with rank $5$.
    \item
    Enqueue the packet itself in the third child of the first child of the root node with rank $6$.
\end{enumerate}

\noindent
We now define the function $\push$, which acts on a path for the topology of a PIFO tree.
We assume each PIFO admits a function $\pifopush: \mathsf{PIFO}(S) \times S \times \mathsf{Rk} \to \mathsf{PIFO}(S)$, which takes a PIFO, an element, and a rank, and returns an updated PIFO with the element enqueued at the given rank.

\begin{definition}
Let $t \in \Topo$.
For all $\mathit{pkt} \in \mathsf{Pkt}$, we define the function $\push: \mathsf{PIFOTree}(t) \times \mathsf{Pkt} \times \mathsf{Path}(t)\to \mathsf{PIFOTree}(t) $ inductively, as follows (types are inferred from context, as before).
\begin{mathpar}
    \inferrule{%
        \pifopush (p, \mathit{pkt}, r) = p'
    }{%
        \push (\mathsf{Leaf}(p), \mathit{pkt}, r) = \mathsf{Leaf}(p')
    }
    \and
    \inferrule{%
        \push (\vec{q}[i], \mathit{pkt}, \mathit{pt}) = q' \and
        \pifopush (p, i, r) = p'
    }{%
        \push (\mathsf{Internal}(\vec{q}, p), \mathit{pkt}, (i, r) :: pt) =
            \mathsf{Internal}(\vec{q}[q'/i], p')
    }
\end{mathpar}
\end{definition}

When the node is a leaf, $\push$ simply enqueues the packet into the PIFO using the given rank.
When the node is internal, it enqueues an index into the appropriate child into its own PIFO, and recurses by calling $\push$ on the appropriate child with an appropriately shortened path (in either order).
If $\push$ succeeds on the child, an altered version of the child is returned; we functionally update the present node to reflect its new index-PIFO and list of PIFO tree children.

\medskip
We model the scheduling algorithm that is run on the tree separately in a \emph{control}, which keeps track of (1) a state from a fixed set $\mathsf{St}$ of states, as well as (2) a \emph{scheduling transaction} that decides on a path (of the right topology) for each incoming packet, while possibly updating the state.

\begin{definition}
Let $t \in \Topo$.
A \emph{control} over $t$ is a triple $(s, q, z)$, where $s \in \mathsf{St}$ is the \emph{current state}, $q$ is a PIFO tree of topology $t$, and $z: \mathsf{St} \times \mathsf{Pkt} \to \mathsf{Path}(t) \times \mathsf{St}$ is a function called the \emph{scheduling transaction}.
We write $\mathsf{Control}(t)$ for the set of controls over $t$.
\end{definition}

\begin{remark}
To recover the presentation of the node-bound scheduling transactions from~\citet{Sivaraman16}, one can just project the return value of $z$ to obtain a rank and child index for each node (updating the state only in the scheduling transaction for the root node).
Conversely, individual scheduling transactions for each node can be glued together into a monolithic one that has the same effect.
For the sake of brevity, we do not expand on these constructions.
\end{remark}

\subsection{Well-Formedness}%
\label{sec:wf}

From the development preceding, it should be clear that the $\pop$ and $\push$ operations do not change a PIFO tree's topology.
There is, however, one more important invariant which these operations maintain, but which we have not yet discussed.
This has to do with the \emph{validity} of references carried by the internal nodes.
Put simply, the $\pop$ and $\push$ operations maintain that if the PIFO of an internal node carries $n$ references to its $m$-th child, then the leaves below that child carry exactly~$n$ packets.
This prevents $\pop$ from ``getting stuck'' as it traverses a non-empty tree looking for a packet.
We now formalize this notion in a typing relation, as follows.

\begin{definition}
Let $S$ be a set, and $p \in \mathsf{PIFO}(S)$.
We write $|p|$ for the size of $p$, which is the number of elements it holds; furthermore, when $s \in S$, we write $|p|_s$ for the number of times $s$ occurs in $p$.

Let $q$ be a PIFO tree.
We write $|q|$ for the \emph{size} of $q$, which is defined as the number of packets enqueued at the leaves.
Formally, $|\cdot|: \mathsf{PIFOTree}(t) \to \mathbb{N}$ is defined for every $t \in \Topo$ by
\begin{mathpar}
    |\mathsf{Leaf}(p)| = |p|
    \and
    |\mathsf{Internal}(\vec{q}, p)| = |\vec{q}[1]| + \cdots + |\vec{q}[n]| \quad (n = |\vec{q}|)
\end{mathpar}
We say that $q$ is \emph{well-formed}, denoted $\vdash q$, if it adheres to the following rules.
\begin{mathpar}
    \inferrule{~}{%
        \vdash \mathsf{Leaf}(p)
    }
    \and
    \inferrule{%
        \forall 1 \leq i \leq |\vec{q}|.\ {\vdash \vec{q}[i]} \wedge {|p|_i = |\vec{q}[i]|}
    }{%
        \vdash \mathsf{Internal}(\vec{q}, p)
    }
\end{mathpar}
\end{definition}

Intuitively, leaves are well-formed, and an internal node is well-formed if, for all legal values of child-index~$i$, the~$i$-th child is itself a well-formed PIFO tree, and the number of times~$i$ occurs in the index-PIFO~$p$ is equal to the number of packets held by the~$i$-th child.

With all of these notions in place, we are ready to formally state the invariants for PIFO trees.
\begin{lemma}%
\label{lemma:sanity-check}
Let $t \in \Topo$, $q \in \mathsf{PIFOTree}(t)$ such that $\vdash q$, and $\mathit{pkt} \in \mathsf{Pkt}$.
The following hold.
\begin{enumerate}[(i)]
    \item
    If $pt \in \mathsf{Path}(t)$, then $\push(q, \mathit{pkt}, pt)$ is well-defined, and $\vdash \push(q, \mathit{pkt}, pt)$.
    \item
    If $|q| > 0$, then $\pop(q)$ is well-defined, and $\vdash q'$ where $\pop(q) = (\mathit{pkt}, q')$.
\end{enumerate}
\end{lemma}
\begin{proof}[Proof sketch]
Both of these follow by (dependent) induction on $t$.
\end{proof}

\section{Limits to Expressiveness}%
\label{SEC:LIMITS}

We have already suggested that some PIFO trees can express more policies than others.
This section shows formally that, in general, PIFO trees with more leaves are more expressive.
In turn, this tells us that a fixed-topology hardware implementation of PIFO trees needs to support a sufficient number of leaves if it is to express policies of practical interest.

Throughout this section, we use word notation for finite sequences of varying length, to contrast with the list notation used for sequences of fixed length.
For instance, we will use juxtaposition to construct lists of packets $pkt_3pkt_2pkt_1$ and write $\cdot$ for the concatenation operator.
The constant $\epsilon$ denotes an empty sequence, and $S^*$ is the set of all finite sequences over elements of a set $S$.
\ifarxiv%
For the sake of brevity, most proofs relating to this section are deferred to \Cref{appendix:proofs-limits}.
\fi%

\subsection{Simulation}
Up to this point, we have been imprecise about what it means for one PIFO tree to replicate the behavior of another, but now we need a formal notion.
To this end, we instantiate a tried and true idea from process algebra~\cite{milner-1971}, and more generally, coalgebra~\cite{rutten-2000}.

Intuitively, we wish to ask whether, given PIFO trees~$q_1$ and~$q_2$, any $\pop$ (resp.\ $\push$) performed on~$q_1$ can be mimicked by a $\pop$ (resp.\ $\push$) on~$q_2$, such that~$q_1$ and~$q_2$ schedule packets identically.

\begin{definition}
Let $t_1, t_2 \in \Topo$.
We call a relation $R \subseteq \mathsf{PIFOTree}(t_1) \times \mathsf{PIFOTree}(t_2)$ a \emph{simulation} if it satisfies the following conditions, for all $pkt \in \mathsf{Pkt}$ and $q_1 \mathrel{R} q_2$:
\begin{enumerate}
    \item
    If $\pop(q_1)$ is undefined, then so is $\pop(q_2)$.
    \item
    If $\pop(q_1) = (pkt, q_1')$, then $\pop(q_2) = (pkt, q_2')$ such that $q_1' \mathrel{R} q_2'$.
    \item
    For all $pt_1 \in \mathsf{Path}(t_1)$, there exists a $pt_2 \in \mathsf{Path}(t_2)$ such that
    \[
        \push(q_1, pkt, pt_2) \; \mathrel{R} \; \push(q_2, pkt, pt_2).
    \]
\end{enumerate}
If such a simulation exists, we say that $q_1$ \emph{is simulated by} $q_2$, and we write $q_1 \preceq q_2$.
\end{definition}
It follows that, if~$c_1$ is a control for~$q_1$, then there exists a control~$c_2$ for~$q_2$ that can make~$q_2$ behave in the same way as~$q_1$.
Conversely, if~$q_1$ is \emph{not} simulated by~$q_2$, then there exists some control~$c_1$ for~$q_1$ whose scheduling decisions cannot be replicated by \emph{any} control for~$q_2$.

\begin{example}
To build intuition, let us construct a simple simulation between two trees of the same topology.
Let ${t \in \Topo}$, and let ${s: \mathsf{Rnk} \to \mathsf{Rnk}}$ be some monotone function on $\mathsf{Rnk}$.
We define $R_s \subseteq \mathsf{PIFOTree}(t) \times \mathsf{PIFOTree}(t)$ as the relation where $q_1 \mathrel{R_s} q_2$ if and only if $q_1$ and $q_2$ agree on the contents of PIFOs in corresponding (leaf or internal) nodes, except that if $r$ is the rank of an item in $q_1$, then $s(r)$ is the rank of that item in $q_2$.
It is then easy to show that $R_s$ is a simulation; in particular, when $pkt \in \mathsf{Pkt}$ and $pt_1 \in \mathsf{Path}(t)$ such that $\push(q_1, pkt, pt_1) = q_1'$, we can apply~$s$ to all the ranks in~$pt_1$ to obtain~${pt_2 \in \mathsf{Path}(t)}$, and note that ${q_1' \;\; R_s \; \push(q_2, pkt, pt_2)}$.
\end{example}

\subsection{Snapshots and Flushes}

Showing that a relation on PIFO trees is a simulation is not always trivial, but typically doable.
In contrast, showing that a PIFO tree $q_1$ \emph{cannot} be simulated by a PIFO tree $q_2$ is a different endeavor altogether.
To argue this formally, we have to find a property that remains true for $q_1$ over a sequence of actions, but is eventually falsified when mimicking those actions on $q_2$.
We now introduce \emph{snapshots} and \emph{flushing} as tools to formalize such invariants, and show how the two relate.

A snapshot tells us about the relative order of packets already fixed by leaf nodes alone, but completely disregards the ordering dictated by the indices in the tree's internal PIFOs.

\begin{definition}
Let $p \in \mathsf{PIFO}(S)$ be a PIFO over some set $S$.
We write $\textsc{flush}(p)$ for the sequence of elements (from $S^*$) retrieved by repeatedly calling $\pifopop$ on the PIFO $p$ until it is empty, with the last $\pifopop$ped element on the far left.
Furthermore, let~$q$ be a PIFO tree.
We inductively define the \emph{snapshot} of~$q$, denoted $\mathsf{snap}(q)$, as a sequence of sequences given by:
\begin{align*}
    \mathsf{snap}(\mathsf{Leaf}(p)) &= [\textsc{flush}(p)] \\
    \mathsf{snap}(\mathsf{Internal}(\vec{q}, p)) &= \mathsf{snap}(\vec{q}[1]) \mathop{+\!+} \cdots \mathop{+\!+} \mathsf{snap}(\vec{q}[n]) \quad (n = |q|)
\end{align*}
To prevent confusion between the different levels of finite sequences, we write $[x]$ to denote a single-element list containing $x$, and denote list concatenation on the upper level by $+\!+$.
\end{definition}

A more delicate concept, $\flush$, \emph{does} take into account the indices enqueued in the tree's internal PIFOs: it yields the sequence obtained by $\pop$ping a well-formed PIFO tree until it is empty.
Keeping with convention, $\flush$ puts the most favorably ranked packet in the PIFO at the right of its output.

\begin{definition}
Let $t \in \Topo$ and let $q \in \mathsf{PIFOTree}(t)$ be a well-formed PIFO tree.
We define $\flush(q)$ as a list of packets by induction on $|q|$, as follows:
\begin{mathpar}
    \inferrule{%
        |q| = 0
    }{%
        \flush(q) = \epsilon
    }
    \and
    \inferrule{%
        |q| > 0 \and
        \pop(q) = (pkt, q')
    }{%
        \flush(q) = \flush(q') \cdot pkt
    }
\end{mathpar}
\end{definition}

\begin{restatable}{lemma}{restateflushinterleave}%
\label{lemma:flush-interleave}
If $q$ is a well-formed PIFO tree, $\flush(q)$ is an interleaving of the lists in $\mathsf{snap}(q)$.
\end{restatable}

It is possible for PIFO trees to disagree on their snapshots, but still be in a simulation---for instance, a PIFO tree could simulate a PIFO tree of the same topology, while permuting the contents of some of the leaves.
However, well-formed PIFO trees that are in simulation must agree on $\flush$.

\begin{restatable}{lemma}{restatesimflush}%
\label{lemma:sim-flush}
Let $q_1$ and $q_2$ be well-formed PIFO trees.
If $q_1 \preceq q_2$, then $\flush(q_1) = \flush(q_2)$.
\end{restatable}

\subsection{Fewer Leaves are Less Expressive}

We now put the formalisms introduced above to work, by constructing, for each topology with~$n$ leaves, a PIFO tree that \emph{cannot} be simulated by any PIFO tree with fewer leaves.
We show that packets in different leaves can always be permuted by a specially chosen sequence of actions, while packets that appear in the same leaf are forced to obey the relative order they have in that leaf.

\begin{definition}
Let $t \in \Topo$.
We write $|t|$ for the \emph{number of leaves} in $t$; formally,
\begin{mathpar}
|{*}| = 1
\and
|\mathsf{Node}(\vec{t})| = |\vec{t}[1]| + \cdots + |\vec{t}[n]| \quad (n = |\vec{t}|)
\end{mathpar}
\end{definition}

\Cref{lemma:flush-interleave} tells us that the $\flush$ of a PIFO tree is an interleaving of the lists in its snapshot.
In particular, when each leaf holds a single element, \emph{any} ordering of the packets in this tree is an interleaving of these snapshots.
Say each leaf indeed holds only one packet.
Can any ordering of packets be achieved?
The following lemma answers this question in the positive, and in the process provides us with a construction to obtain a tree that describes any permutation.

\begin{restatable}{lemma}{restateconstructperm}%
\label{lemma:construct-perm}
Let $t \in \Topo$, and let $P = \{ pkt_1, \dots, pkt_{|t|} \}$ be a set of distinct packets.
For each permutation $\pi$ on $P$, there exists a well-formed PIFO tree $q_\pi \in \mathsf{PIFOTree}(t)$ such that
\begin{mathpar}
\flush(q_\pi) = \pi(pkt_{|t|}) \cdots \pi(pkt_1)
\and
\mathsf{snap}(q_\pi) = [pkt_1, \ldots, pkt_{|t|}]
\end{mathpar}
\end{restatable}

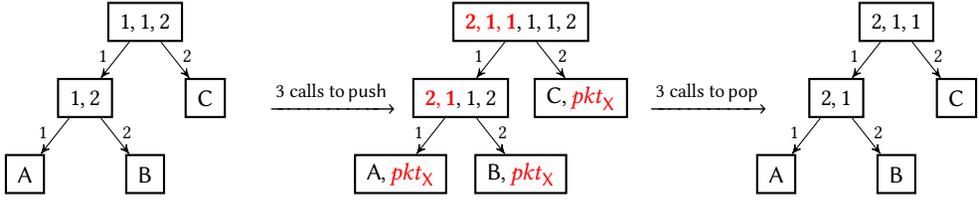
\begin{figure}
  \begin{center}
  \begin{tikzpicture}[->, >=stealth', auto, bend angle=10]
  \tikzstyle{dot}=[state, inner sep=1.3pt, minimum size=0pt, draw=black, fill=black]
  \tikzstyle{box}=[inner sep=0pt, minimum height=0.5cm, draw=black, thick]
  \small
    \node[box] (root) {$\phantom{b}1,1,2\phantom{b}$};
    \node[box] (N2) [below of=root, xshift=8mm] {$\phantom{b}\mathsf{C}\phantom{b}$};
    \node[box] (N1) [below of=root, xshift=-8mm] {$\phantom{b}1,2\phantom{b}$};
    \node[box] (N00) [below of=N1, xshift=-8mm] {$\phantom{b}\mathsf{A}\phantom{b}$};
    \node[box] (N01) [below of=N1, xshift=8mm] {$\phantom{b}\mathsf{B}\phantom{b}$};
    \path (root) edge node[swap, pos=.65, xshift=2pt, yshift=-2pt] {$\scriptstyle 1$} (N1);
    \path (root) edge node[pos=.65, xshift=-2pt, yshift=-2pt] {$\scriptstyle 2$} (N2);
    \path (N1) edge node[swap, pos=.65, xshift=2pt, yshift=-2pt] {$\scriptstyle 1$} (N00);
    \path (N1) edge node[pos=.65, xshift=-2pt, yshift=-2pt] {$\scriptstyle 2$} (N01);%
    \node [right of=root, node distance=25mm, yshift=-10mm] {$\xrightarrow{3 \text{ calls to} \push}$};
    \node[box] (root') [right of=root, node distance=50mm] {$\phantom{b}{\bf\color{red}{2,1,1}},1,1,2\phantom{b}$};
    \node[box] (N2') [below of=root', xshift=8mm] {$\phantom{b}\mathsf{C},{\color{red}\X}\phantom{b}$};
    \node[box] (N1') [below of=root', xshift=-8mm] {$\phantom{b}{\bf\color{red}{2,1}},1,2\phantom{b}$};
    \node[box] (N00') [below of=N1', xshift=-8mm] {$\phantom{b}\mathsf{A},{\color{red}\X}\phantom{b}$};
    \node[box] (N01') [below of=N1', xshift=8mm] {$\phantom{b}\mathsf{B},{\color{red}\X}\phantom{b}$};
    \path (root') edge node[swap, pos=.65, xshift=2pt, yshift=-2pt] {$\scriptstyle 1$} (N1');
    \path (root') edge node[pos=.65, xshift=-2pt, yshift=-2pt] {$\scriptstyle 2$} (N2');
    \path (N1') edge node[swap, pos=.65, xshift=2pt, yshift=-2pt] {$\scriptstyle 1$} (N00');
    \path (N1') edge node[pos=.65, xshift=-2pt, yshift=-2pt] {$\scriptstyle 2$} (N01');%
    \node [right of=root', node distance=25mm, yshift=-10mm] {$\xrightarrow{3 \text{ calls to} \pop}$};
    \node[box] (root'') [right of=root', node distance=50mm] {$\phantom{b}2,1,1\phantom{b}$};
    \node[box] (N2'') [below of=root'', xshift=8mm] {$\phantom{b}\mathsf{C}\phantom{b}$};
    \node[box] (N1'') [below of=root'', xshift=-8mm] {$\phantom{b}2,1\phantom{b}$};
    \node[box] (N00'') [below of=N1'', xshift=-8mm] {$\phantom{b}\mathsf{A}\phantom{b}$};
    \node[box] (N01'') [below of=N1'', xshift=8mm] {$\phantom{b}\mathsf{B}\phantom{b}$};
    \path (root'') edge node[swap, pos=.65, xshift=2pt, yshift=-2pt] {$\scriptstyle 1$} (N1'');
    \path (root'') edge node[pos=.65, xshift=-2pt, yshift=-2pt] {$\scriptstyle 2$} (N2'');
    \path (N1'') edge node[swap, pos=.65, xshift=2pt, yshift=-2pt] {$\scriptstyle 1$} (N00'');
    \path (N1'') edge node[pos=.65, xshift=-2pt, yshift=-2pt] {$\scriptstyle 2$} (N01'');
  \end{tikzpicture}
  \end{center}
  \caption{Effecting permutations in $2|t|$ steps.}%
  \label{fig:effecting-permutations}
\end{figure}

As it happens, PIFO trees as constructed in this way for a fixed topology are closely related, in the sense that each $q_\pi$ can be transformed into each $q_{\pi'}$ (up to simulation) via a specific sequence of operations.
The idea is to use $\push$ to append the contents of the PIFO of each internal node in $q_{\pi'}$ to the PIFO of the corresponding node in $q_\pi$.
The packets pushed will be the dummy packets $\X$, which will be scheduled at the head of each leaf.
A sequence of $\pop$s will then clear these dummy packets, as well as the original contents of the internal node PIFOs.
The tree that results is very similar to $q_\pi$, except that the ranks used in its internal PIFOs may have shifted by some constant.

As an example, consider the PIFO tree on the left in \Cref{fig:effecting-permutations}, which flushes to $\mathsf{ABC}$.
We can push $\X$ three times to obtain the PIFO tree in the middle, where the added contents are given in \textcolor{red}{red}.
Calling $\pop$ three times yields the PIFO tree on the right, which itself flushes to $\mathsf{CBA}$.

\begin{restatable}{lemma}{restatepermutation}%
\label{lemma:permutation}
Let $t \in \Topo$, let $P = \{ pkt_1, \dots, pkt_{|t|} \} \subseteq \mathsf{Pkt}$, and let $\pi,\pi'$ be permutations on~$P$.
Further, let $\X$ be a packet not occurring in $P$.
There exists some $q \in \mathsf{PIFOTree}(t)$ and a sequence of~$|t|$ $\push$es of $\X$ followed by $|t|$ $\pop$s that transforms $q_\pi$ into $q$, such that $q \preceq q_{\pi'}$.
\end{restatable}

\begin{remark}
\Cref{lemma:construct-perm,lemma:permutation} can be generalized to PIFO trees with multiple elements enqueued at each leaf.
In a more lax notion of simulation, where the simulating tree may respond with more than one $\push$ or $\pop$ operation, this can be used to interrelate the semantics of PIFO trees with differing topologies.
\ifarxiv%
We refer interested readers to \Cref{appendix:more-general-simulation} for more details.
\fi%
\end{remark}

With this transformation lemma in place, we are ready to formally state and prove a critical theorem: a PIFO tree can never be simulated by another PIFO tree with fewer leaves.
Contrapositively, if a PIFO tree $q_1$ is simulated by a PIFO tree $q_2$, then $q_2$ has at least as many leaves as~$q_1$.

\begin{restatable}{theorem}{restatesimulationsvsleaves}%
\label{thm:simulations-vs-leaves}
Let ${t_1, t_2 \in \Topo}$ with $|t_1| > |t_2|$.
For all $q_1 \in \mathsf{PIFOTree}(t_1)$ and $q_2 \in \mathsf{PIFOTree}(t_2)$ such that $\vdash q_1$ and $\vdash q_2$, we have that $q_1 \preceq q_2$ does \emph{not} hold.
\end{restatable}
\begin{proof}[Proof sketch]
Take a set of distinct packets $P = \{ pkt_1, \dots, pkt_{|t_1|} \}$.
We first consider the case where $q_1 = q_\mathsf{id}$, in which $\mathsf{id}$ is the identity permutation.
If $q_1$ is simulated by $q_2 \in \mathsf{PIFOTree}(t_2)$, then $q_2$ has a leaf with at least two packets, $pkt_i$ and $pkt_j$.
By \Cref{lemma:permutation}, we can then permute those packets in $q_1$ to obtain a PIFO tree $q_1'$, using a sequence of $\push$ and $\pop$ operations that only involve a dummy packet $\X$.
Any way of applying $\push$ and $\pop$ operations to $q_2$ using this dummy packet will never change the position of $pkt_i$ and $pkt_j$, because they are in the same leaf; hence, $q_2$ cannot simulate $q_1$.
The more general case can be reduced to the above for well-formed trees, by simply $\pop$ping $q_1$ until it is empty, and then $\push$ing packets to turn it into $q_\mathsf{id}$.
\end{proof}

\section{Embedding and Simulation}%
\label{SEC:EMBEDDING} %

In the previous section, we formalized what it means for one tree to simulate the queuing behavior of another, and showed that the number of leaves is important: a PIFO tree \emph{cannot} simulate a PIFO tree with more leaves.
We now turn our attention to the converse: when can one PIFO tree simulate another?
As explained, we will focus primarily on the \emph{topology} of PIFO trees.
\ifarxiv%
In this section, we provide proof sketches to guide intution; more extensive proofs are available in \Cref{appendix:proofs-embedding}.
\fi%

In a nutshell, the results in this section tell us that a hardware implementation of PIFO trees may fix a certain topology.
As long as the topologies of PIFO trees designed by the user embed in this topology (see below), the hardware implementation will be able to replicate their behavior.

\subsection{Embedding}

We start by developing a notion of \emph{embedding} at the level of PIFO tree topologies, and show that if $t_1 \in \Topo$ embeds inside $t_2 \in \Topo$, then PIFO trees over $t_1$ can be simulated by PIFO trees over $t_2$.
Intuitively, an embedding maps nodes of one topology to nodes of the other in a way that allows internal nodes to designate (new) intermediate nodes as their children.
To make this precise, we need a way to refer to nodes and leaves internal to a topology, which we formalize as follows.

\begin{definition}[Addresses]
Let $t \in \Topo$.
We write $\mathsf{Addr}(t)$ for the set of \emph{(node) addresses} in~$t$, which is defined as the smallest subset of $\mathbb{N}^*$ satisfying the rules
\begin{mathpar}
    \inferrule{~}{%
        \epsilon \in \mathsf{Addr}(t)
    }
    \and
    \inferrule{%
        n \in \mathbb{N} \\
        \vec{t} \in \Topo^n \\
        \addr \in \mathsf{Addr}(\vec{t}[i]) \\
        1 \leq i \leq n
    }{%
        i \cdot \addr \in \mathsf{Addr}(\mathsf{Node}(\vec{t}))
    }
\end{mathpar}
Intuitively, $\epsilon$ addresses the root, while $i \cdot \alpha$ refers to the node pointed to by $\alpha$ in the $i$-th subtree.

When $\addr \in \mathsf{Addr}(t)$, we write $t/\addr$ for the subtree reached by $\addr$, defined inductively by
\begin{mathpar}
    t/\epsilon = t
    \and
    \mathsf{Node}(\vec{t})/(i \cdot \addr) = \vec{t}[i]/\addr
\end{mathpar}
This defines a total function, because if $\addr = i \cdot \addr'$, then $t$ is necessarily of the form $\mathsf{Node}(\vec{t})$ for $\vec{t} \in \Topo^n$ such that $1 \leq i \leq n$, and $\addr' \in \mathsf{Addr}(\vec{t}[i])$.
\end{definition}

Note that this way of addressing nodes is compatible with the intuition of \emph{ancestry}: if $\addr$ and $\addr'$ are addresses in $t$, then $\addr$ points to an ancestor of the node referred to by $\addr'$ precisely when $\addr$ is a prefix of $\addr'$.
This guides the definition of embedding, as follows.

\begin{definition}[Embedding]
Let $t_1, t_2 \in \Topo$.
A \emph{(homomorphic) embedding} of $t_1$ in $t_2$ is an injective map $f: \mathsf{Addr}(t_1) \to \mathsf{Addr}(t_2)$ such that, for all $\addr, \addr' \in \mathsf{Addr}(t_1)$, three things hold:
\begin{enumerate}
    \item
    $f$ maps the root of $t_1$ to the root of $t_2$, i.e., if $\addr = \epsilon$ then $f(\addr) = \epsilon$.
    \item
    $f$ maps leaves of $t_1$ to leaves of $t_2$, i.e., if $t_1/\addr = *$, then $t_2/f(\addr) = *$.
    \item
    $f$ respects ancestry, i.e., $\addr$ is a prefix of $\addr'$ iff $f(\addr)$ is a prefix of $f(\addr')$.
\end{enumerate}
\end{definition}

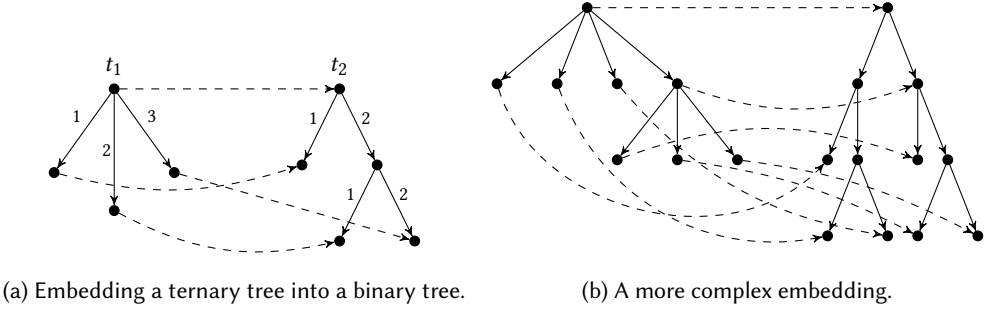
\begin{figure}
    \begin{subfigure}[b]{0.45\textwidth}
        \centering
        \begin{tikzpicture}[->, >=stealth', auto, bend angle=20]
        \tikzstyle{dot}=[state, inner sep=1.3pt, minimum size=0pt, draw=black, fill=black]
        \small
        \node[dot, label=above:{$t_1$}] (R) {};
        \node[dot] (C2) [below of=R, yshift=-6mm] {};
        \node[dot] (C1) [left of=C2, yshift=5mm, node distance=8mm] {};
        \node[dot] (C3) [right of=C2, yshift=5mm, node distance=8mm] {};
        \path (R) edge node[swap, xshift=2pt] {$\scriptstyle 1$} (C1);
        \path (R) edge node[swap, xshift=2pt] {$\scriptstyle 2$} (C2);
        \path (R) edge node[xshift=-2pt] {$\scriptstyle 3$} (C3);
        \node[dot, label=above:{$t_2$}] (eps) [right of=R, node distance=3cm]{};
        \node[dot] (N0) [below of=eps, xshift=-5mm] {};
        \node[dot] (N1) [below of=eps, xshift=5mm] {};
        \node[dot] (N10) [below of=N1, xshift=-5mm] {};
        \node[dot] (N11) [below of=N1, xshift=5mm] {};
        \path (eps) edge node[swap, xshift=2pt, yshift=-2pt] {$\scriptstyle 1$} (N0);
        \path (eps) edge node[xshift=-2pt, yshift=-2pt] {$\scriptstyle 2$} (N1);
        \path (N1) edge node[swap, xshift=2pt, yshift=-2pt] {$\scriptstyle 1$} (N10);
        \path (N1) edge node[xshift=-2pt, yshift=-2pt] {$\scriptstyle 2$} (N11);
        \path (R) edge[dashed] (eps);
        \path (C2) edge[dashed, bend right] (N10);
        \path (C3) edge[dashed] (N11);
        \path (C1) edge[dashed, bend angle=16, bend right] (N0);

        \end{tikzpicture}
        \caption{Embedding a ternary tree into a binary tree.}%
        \label{fig:example-embedding-a}
    \end{subfigure}
    \begin{subfigure}[b]{0.5\textwidth}
        \centering
        \begin{tikzpicture}[->, >=stealth', auto, bend angle=20]
        \tikzstyle{dot}=[state, inner sep=1.3pt, minimum size=0pt, draw=black, fill=black]
        \small
        \node[dot] (R) {};
        \node[dot] (C2) [below of=R, xshift=-4mm] {};
        \node[dot] (C1) [left of=C2, node distance=8mm] {};
        \node[dot] (C3) [below of=R, xshift=4mm] {};
        \node[dot] (C4) [right of=C3, node distance=8mm] {};
        \node[dot] (C42) [below of=C4] {};
        \node[dot] (C41) [left of=C42, node distance=8mm] {};
        \node[dot] (C43) [right of=C42, node distance=8mm] {};
        \path (R) edge node[swap, near end, xshift=2pt, yshift=-2pt] {} (C1);
        \path (R) edge node[swap, near end, xshift=2pt, yshift=-2pt] {} (C2);
        \path (R) edge node[near end, xshift=-2pt, yshift=-2pt] {} (C3);
        \path (R) edge node[near end, xshift=-2pt, yshift=-2pt] {} (C4);
        \path (C4) edge node[swap, xshift=2pt] {} (C41);
        \path (C4) edge node[swap, xshift=2pt] {} (C42);
        \path (C4) edge node[xshift=-2pt] {} (C43);

        \node[dot] (eps) [right of=R, node distance=4cm]{};
        \node[dot] (N0) [below of=eps, xshift=-8mm] {};
        \node[dot] (N00) [below of=N0, xshift=-4mm] {};
        \node[dot] (N01) [below of=N0, xshift=0mm] {};
        \node[dot] (N1) [below of=eps, xshift=0mm] {};

        \node[dot, fill=lightgray] (N2) [below of=eps, xshift=8mm] {};
        \node[dot, fill=lightgray] (N20) [below of=N2, xshift=0mm] {};
        \node[dot, fill=lightgray] (N21) [below of=N2, xshift=4mm] {};
        \path (eps) edge node[xshift=-2pt, yshift=-2pt] {} (N2);
        \path (N2) edge node[xshift=-2pt, yshift=-2pt] {} (N20);
        \path (N2) edge node[xshift=2pt, yshift=-2pt] {} (N21);

        \node[dot] (N10) [below of=N1, xshift=0mm] {};
        \node[dot] (N11) [below of=N1, xshift=4mm] {};
        \node[dot] (N110) [below of=N11, xshift=-4mm] {};
        \node[dot] (N111) [below of=N11, xshift=4mm] {};
        \path (N1) edge node[swap, xshift=2pt, yshift=-2pt] {} (N10);
        \path (N1) edge node[xshift=-2pt, yshift=-2pt] {} (N11);
        \path (N11) edge node[swap, xshift=2pt, yshift=-2pt] {} (N110);
        \path (N11) edge node[xshift=-2pt, yshift=-2pt] {} (N111);

        \node[dot] (N010) [below of=N01, xshift=-4mm] {};
        \node[dot] (N011) [below of=N01, xshift=4mm] {};
        \path (eps) edge node[swap, xshift=2pt, yshift=-2pt] {} (N0);
        \path (N0) edge node[swap, xshift=2pt, yshift=-2pt] {} (N00);
        \path (N0) edge node[xshift=-2pt, yshift=-2pt] {} (N01);
        \path (eps) edge node[xshift=-2pt, yshift=-2pt] {} (N1);
        \path (N01) edge node[swap, xshift=2pt, yshift=-2pt] {} (N010);
        \path (N01) edge node[xshift=-2pt, yshift=-2pt] {} (N011);
        \path (R) edge[dashed] (eps);
        \path (C4) edge[dashed, bend right] (N1);
        \path (C2) edge[dashed, bend angle=40, bend right] (N010);
        \path (C3) edge[dashed, bend angle=-10, bend right] (N011);
        \path (C1) edge[dashed, bend angle=60, bend right] (N00);
        \path (C41) edge[dashed, bend left] (N10);
        \path (C42) edge[dashed, bend angle=10, bend left] (N110);
        \path (C43) edge[dashed, bend angle=10, bend left] (N111);
        \end{tikzpicture}
        \caption{A more complex embedding.}%
        \label{fig:example-embedding-b}
    \end{subfigure}
\caption{Two examples of embedding.}%
\end{figure}

\noindent
For example, in \Cref{fig:example-embedding-a}, the ternary tree $t_1$ embeds inside the binary tree $t_2$ via the embedding $f$:
\begin{mathpar}
    f(\epsilon) = \epsilon
    \and
    f(1) = 1
    \and
    f(2) = 21
    \and
    f(3) = 22
\end{mathpar}
We do not explicate the embedding shown in \Cref{fig:example-embedding-b}, and we drop the child-indices to lighten the presentation, but observe that the three conditions of an embedding hold.
In general, neither the source nor the target tree needs to be regular-branching so long as these conditions are obeyed.
Further, the embedding is not itself required to be surjective on leaves---\Cref{fig:example-embedding-b} has leaves in the target tree that are not images of leaves in the source tree---but an embedding that is surjective on leaves will shortly become useful.
We construct a leaf-surjective embedding from our injective embedding by ignoring untargeted fragments of the target tree, e.g., the gray subtree in \Cref{fig:example-embedding-b}.

Given an embedding of a topology $t_1$ into another topology $t_2$, we can obtain embeddings of subtrees of $t_1$ into subtrees of $t_2$; this gives us a way to superimpose the inductive structure of topologies onto embeddings, which we will rely on for the remainder of this section.

\begin{restatable}{lemma}{restatedecomposeembedding}%
\label{lemma:decompose-embedding}
Let $t_1, t_2 \in \Topo$, and let $f: \mathsf{Addr}(t_1) \to \mathsf{Addr}(t_2)$ be an embedding.
The following hold:
    (1)~if $t_1 = *$, then $t_2 = *$ as well; and
    (2)~if $t_1 = \mathsf{Node}(\vec{t}_1)$, then there exists for $1 \leq i \leq |\vec{t}_1|$ an embedding $f_i$ of $t_1/i$ inside $t_2/f(i)$ satisfying $f(i \cdot \addr) = f(i) \cdot f_i(\addr)$.
\end{restatable}

\subsection{Lifting Embeddings}

Next, we develop a way to lift these embeddings so that they range over PIFO trees.
Specifically, given an embedding $f: \mathsf{Addr}(t_1) \to \mathsf{Addr}(t_2)$, we lift it into a map $\widehat{f}:~\mathsf{PIFOTree}(t_1) \to \mathsf{PIFOTree}(t_2)$.
Intuitively, if $q$ is a PIFO tree over $t_1$, $\widehat{f}(q)$ is a PIFO tree over $t_2$ that simulates $q$.
The map places packets at the leaves of the input PIFO tree at the corresponding leaves of the output PIFO tree, and populates the PIFOs at the internal nodes of the output tree in a way that preserves $\push$ and $\pop$.

\begin{definition}
Let $t_1, t_2 \in \Topo$, and let $f$ be an embedding of $t_1$ inside $t_2$.
We make $f$ surjective on leaves by ignoring untargeted fragments of $t_2$.
We then lift $f$ to a map $\widehat{f}$ from $\mathsf{PIFOTree}(t_1)$ to $\mathsf{PIFOTree}(t_2)$ by recursion on $t_1$.

\begin{itemize}
    \item
    In the base, where we have $t_1 = *$, we choose $\widehat{f}(q) = q$.
    This is well-defined by \Cref{lemma:decompose-embedding}.
    \item
    In the recursive step, let ${t_1 = \mathsf{Node}(\vec{t}_1)}$, ${n = |\vec{t}_1|}$, ${q = \mathsf{Internal}(\vec{q}, p)}$, ${p \in \mathsf{PIFO}(\{1, \dots, n\})}$, and, for ${1 \leq i \leq n}$, note that ${\vec{q}[i] \in \mathsf{PIFOTree}(\vec{t}_1[i])}$.
    For each prefix $\addr$ of $f(i)$ for some $1 \leq i \leq n$, we construct $\widehat{f}(q)_{\addr} \in \mathsf{PIFOTree}(t_2/\addr)$, by an inner \emph{inverse} recursion on $\addr$, starting from $f(i)$ for $1 \leq i \leq n$ and working our way to $\epsilon$.
    This eventually yields $\widehat{f}(q) = \widehat{f}(q)_\epsilon \in \mathsf{PIFOTree}(t_2/\epsilon) = \mathsf{PIFOTree}(t_2)$:
\begin{itemize}
    \item
    In the inner base, $\addr = f(i)$ for some $1 \leq i \leq n$.
    We choose $\widehat{f}(q)_{\addr} = \widehat{f_i}(\vec{q}[i])$, where~$f_i$ embeds $t_1/i$ inside $t_2/f(i)$, as obtained through \Cref{lemma:decompose-embedding}.
    This is well-defined because $f$ is injective, and because $\vec{q}[i] \in \mathsf{PIFOTree}(\vec{t}_1/i)$, as $\vec{t}_1/i$ is a subtree of~$t_1$.
    \item
    In the inner recursive step, $\addr$ points to an internal node of $t_2$ with, say, $m$ children.
    For all $1 \leq j \leq m$, $\addr \cdot j$ is a prefix of some $f(i)$.
    By recursion, we know $\widehat{f}(q)_{\addr \cdot j} \in \mathsf{PIFOTree}(t_2/(\addr \cdot j))$ exists.
    We create a new PIFO $p_{\addr}$ by replacing the indices $i$ in $p$ (found at the root of $q$, see above) by $1 \leq j \leq m$ such that $\addr \cdot j$ is a prefix of $f(i)$ if such a $j$ exists, and removing them otherwise.
    Finally, we choose:
    $\forall 1 \leq i \leq m.\ \vec{q}_{f,\addr}[j] = \widehat{f}(q_{\addr \cdot j})$
    and
    $\widehat{f}(q)_{\addr} = \mathsf{Internal}(\vec{q}_{f,\addr}, p_{\addr})$.
    By construction, we then have that $\widehat{f}(q)_{\addr} \in \mathsf{PIFOTree}(t_2/\addr)$.
\end{itemize}
\end{itemize}
\end{definition}

\noindent
To build intuition, let us revisit the embedding $f$ we showed in \Cref{fig:example-embedding-a} and lift it to operate on PIFO trees.
The upper row of \Cref{fig:lifting} shows this lifting at the point when we are computing $\widehat{f}$ for the root node, and we already have at hand the liftings of the form $\widehat{f_i}(s_i)$ for the root's children $s_i$.
We present this in two steps.
First, we remove indices of children that no longer exist below us (\fbox{$3,1,2,2,3,1,2$} becomes \fbox{$3,2,2,3,2$} because we remove $1$s).
Second, we renumber indices for the new address regime (\fbox{$3,1,2,2,3,1,2$} becomes \fbox{$2,1,2,2,2,1,2$} and \fbox{$3,2,2,3,2$} becomes \fbox{$2,1,1,2,1$}).

\begin{figure}
    \begin{subfigure}[b]{\textwidth}
        \centering
        \begin{tikzpicture}[->, >=stealth', auto, bend angle=10]
        \tikzstyle{dot}=[state, inner sep=1.3pt, minimum size=0pt, draw=black, fill=black]
        \tikzstyle{box}=[inner sep=0pt, minimum size=0pt, draw=black]
        \small
        \node[box] (R) {\fbox{$3,1,2,2,3,1,2$}};
        \node[dot] (C2) [below of=R, yshift=-6mm, label=below:{$s_2$}] {};
        \node[dot] (C1) [left of=C2, yshift=5mm, node distance=8mm, label=below:{$s_1$}] {};
        \node[dot] (C3) [right of=C2, yshift=5mm, node distance=8mm, label=below:{$s_3$}] {};
        \path (R) edge node[swap, near end, xshift=2pt, yshift=-2pt] {$\scriptstyle 1$} (C1);
        \path (R) edge node[swap, near end, xshift=2pt, yshift=-2pt] {$\scriptstyle 2$} (C2);
        \path (R) edge node[near end, xshift=-2pt, yshift=-2pt] {$\scriptstyle 3$} (C3);
        \node[box] (eps) [right of=R, node distance=46mm] {\fbox{$3,1,2,2,3,1,2$}};
        \node[dot] (N0) [below of=eps, xshift=-8mm, label=below:{$\widehat{f_1}(s_1)$}] {};
        \node[box] (N1) [below of=eps, xshift=8mm] {\fbox{$3,2,2,3,2$}};
        \node[dot] (N10) [below of=N1, xshift=-8mm, label=below:{$\widehat{f_2}(s_2)$}] {};
        \node[dot] (N11) [below of=N1, xshift=8mm, label=below:{$\widehat{f_3}(s_3)$}] {};
        \path (eps) edge node[swap, pos=.65, xshift=2pt, yshift=-2pt] {$\scriptstyle 1$} (N0);
        \path (eps) edge node[xshift=-2pt, yshift=-2pt] {$\scriptstyle 2$} (N1);
        \path (N1) edge node[swap, xshift=2pt, yshift=-2pt] {$\scriptstyle 1$} (N10);
        \path (N1) edge node[xshift=-2pt, yshift=-2pt] {$\scriptstyle 2$} (N11);
        \node[box] (epsA) [right of=eps, node distance=46mm] {\fbox{$2,1,2,2,2,1,2$}};
        \node[dot] (N0A) [below of=epsA, xshift=-8mm, label=below:{$\widehat{f_1}(s_1)$}] {};
        \node[box] (N1A) [below of=epsA, xshift=8mm] {\fbox{$2,1,1,2,1$}};
        \node[dot] (N10A) [below of=N1A, xshift=-8mm, label=below:{$\widehat{f_2}(s_2)$}] {};
        \node[dot] (N11A) [below of=N1A, xshift=8mm, label=below:{$\widehat{f_3}(s_3)$}] {};
        \path (epsA) edge node[swap, pos=.65, xshift=2pt, yshift=-2pt] {$\scriptstyle 1$} (N0A);
        \path (epsA) edge node[xshift=-2pt, yshift=-2pt] {$\scriptstyle 2$} (N1A);
        \path (N1A) edge node[swap, xshift=2pt, yshift=-2pt] {$\scriptstyle 1$} (N10A);
        \path (N1A) edge node[xshift=-2pt, yshift=-2pt] {$\scriptstyle 2$} (N11A);
        \node [right of=R, node distance=23mm, yshift=-5mm] {$\Imp$};
        \node [right of=eps, node distance=23mm, yshift=-5mm] {$\Imp$};
        \end{tikzpicture}
    \end{subfigure}\vspace{0.4cm}
    \begin{subfigure}[b]{\textwidth}
        \centering
        \begin{tikzpicture}[->, >=stealth', auto, bend angle=10]
        \tikzstyle{dot}=[state, inner sep=1.3pt, minimum size=0pt, draw=black, fill=black]
        \tikzstyle{box}=[inner sep=0pt, minimum size=0pt, draw=black]
        \small
        \node[box] (R) {\fbox{$3,1,2,2,{\color{red}2},3,1,2$}};
        \node[dot] (C2) [below of=R, yshift=-6mm, label=below:{$s_2$}] {};
        \node[dot] (C1) [left of=C2, yshift=5mm, node distance=8mm, label=below:{$s_1$}] {};
        \node[dot] (C3) [right of=C2, yshift=5mm, node distance=8mm, label=below:{$s_3$}] {};
        \path (R) edge node[swap, near end, xshift=2pt, yshift=-2pt] {$\scriptstyle 1$} (C1);
        \path (R) edge node[swap, near end, xshift=2pt, yshift=-2pt] {$\scriptstyle 2$} (C2);
        \path (R) edge node[near end, xshift=-2pt, yshift=-2pt] {$\scriptstyle 3$} (C3);
        \node[box] (eps) [right of=R, node distance=46mm] {\fbox{$3,1,2,2,{\color{red}2},3,1,2$}};
        \node[dot] (N0) [below of=eps, xshift=-8mm, label=below:{$\widehat{f_1}(s_1)$}] {};
        \node[box] (N1) [below of=eps, xshift=8mm] {\fbox{$3,2,2,{\color{red}2},3,2$}};
        \node[dot] (N10) [below of=N1, xshift=-8mm, label=below:{$\widehat{f_2}(s_2)$}] {};
        \node[dot] (N11) [below of=N1, xshift=8mm, label=below:{$\widehat{f_3}(s_3)$}] {};
        \path (eps) edge node[swap, pos=.65, xshift=2pt, yshift=-2pt] {$\scriptstyle 1$} (N0);
        \path (eps) edge node[xshift=-2pt, yshift=-2pt] {$\scriptstyle 2$} (N1);
        \path (N1) edge node[swap, xshift=2pt, yshift=-2pt] {$\scriptstyle 1$} (N10);
        \path (N1) edge node[xshift=-2pt, yshift=-2pt] {$\scriptstyle 2$} (N11);
        \node[box] (epsA) [right of=eps, node distance=46mm] {\fbox{$2,1,2,2,{\color{red}2},2,1,2$}};
        \node[dot] (N0A) [below of=epsA, xshift=-8mm, label=below:{$\widehat{f_1}(s_1)$}] {};
        \node[box] (N1A) [below of=epsA, xshift=8mm] {\fbox{$2,1,1,{\color{red}1},2,1$}};
        \node[dot] (N10A) [below of=N1A, xshift=-8mm, label=below:{$\widehat{f_2}(s_2)$}] {};
        \node[dot] (N11A) [below of=N1A, xshift=8mm, label=below:{$\widehat{f_3}(s_3)$}] {};
        \path (epsA) edge node[swap, pos=.65, xshift=2pt, yshift=-2pt] {$\scriptstyle 1$} (N0A);
        \path (epsA) edge node[xshift=-2pt, yshift=-2pt] {$\scriptstyle 2$} (N1A);
        \path (N1A) edge node[swap, xshift=2pt, yshift=-2pt] {$\scriptstyle 1$} (N10A);
        \path (N1A) edge node[xshift=-2pt, yshift=-2pt] {$\scriptstyle 2$} (N11A);
        \node [right of=R, node distance=23mm, yshift=-5mm] {$\Imp$};
        \node [right of=eps, node distance=23mm, yshift=-5mm] {$\Imp$};
        \end{tikzpicture}
    \end{subfigure}
    \caption{Lifting an embedding (above), and preserving $\push$ across a lifted embedding (below).
    The source and target topologies as shown in \Cref{fig:example-embedding-a}.
    The new indices inserted are in {\color{red}red}.}%
    \label{fig:lifting}
\end{figure}
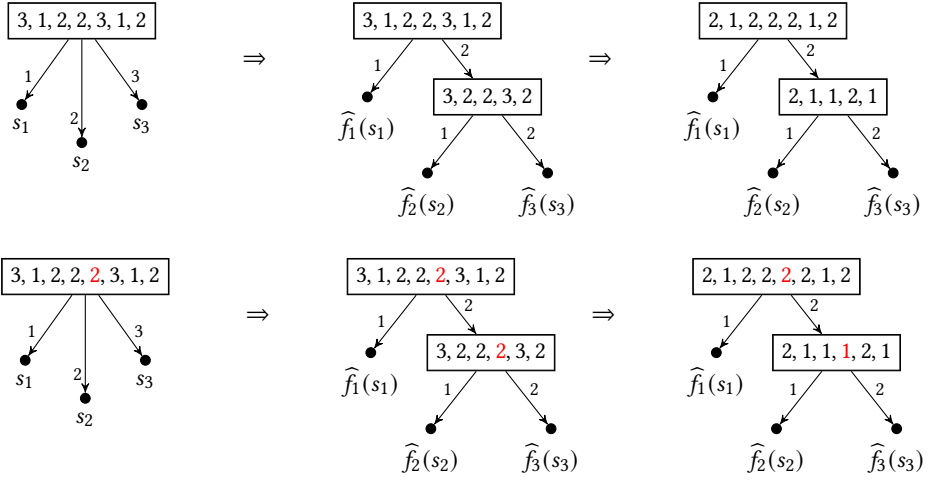

As a sanity check, we verify that a lifted embedding preserves well-formedness of PIFO trees.
The proof of the following property also serves as a small model for the proofs about lifted embeddings that follow, which all proceed by an induction that matches the recursive structure above.

\begin{restatable}{lemma}{restateliftedembeddingwellformed}
Let $t_1, t_2 \in \Topo$, and suppose $f$ embeds $t_1$ inside $t_2$.
If ${\vdash q}$, then ${\vdash \widehat{f}(q)}$.
\end{restatable}
\begin{proof}[Proof sketch]
By induction on $t_1$.
In the base, where $t_1 = *$, this holds because $\widehat{f}(q) = q$.
Next, let $\mathsf{Node}(\vec{t}_1)$ and $n = |\vec{t}_1|$ and perform inverse induction on the prefixes of $f(i)$ for $1 \leq i \leq n$.
Show, more generally, that (1)~if $\addr$ is such a prefix, then $\vdash \widehat{f}(q)_{\addr}$, and (2)~$|\widehat{f}(q)_{\addr}|$ is equal to the sum of $|p|_i$ where $1 \leq i \leq n$ and $\addr$ is a prefix of $f(i)$.
Instantiating (1) with $\addr = \epsilon$ implies the claim.
\end{proof}

\subsection{Preserving \texorpdfstring{$\pop$}{pop}}

Simulation requires that the simulating tree can always be $\pop$ped when the simulated tree can be $\pop$ped.
If $q$ is well-formed and not empty, this must also be the case for $\widehat{f}(q)$---and hence if $q$ can be $\pop$ped so can $\widehat{f}(q)$ in this case.
In fact, this is true regardless of well-formedness.

\begin{restatable}{lemma}{restatepreservedefined}%
\label{lemma:preserve-defined}
Let $t_1, t_2 \in \Topo$, and let $f$ embed $t_1$ inside $t_2$.
If $\pop(q)$ is defined, then so is $\pop(\widehat{f}(q))$.
\end{restatable}
\begin{proof}[Proof sketch]
By induction on $t_1$.
In the base, where $t_1 = *$, the claim is trivial because $\widehat{f}(q) = q$.
For the inductive step, let $t_1 = \mathsf{Node}(\vec{t}_1)$ and $n = |\vec{t}_1|$.
Write $q = \mathsf{Internal}(\vec{q}, p)$.
Since $\pop(q)$ is defined, $p$ cannot be empty; write $i$ for the index at the head of $p$.
We then prove that $\pop(\widehat{f}(q)_{\addr})$ is defined for all prefixes $\addr$ of $f(i)$ by inverse induction on $\addr$; when $\addr = \epsilon$, this implies the claim.
\end{proof}

We furthermore need to show that if $\pop$ping $q$ yields a packet $pkt$ and a new tree $q'$, then $\pop$ping $\widehat{f}(q)$ gives us the same packet $pkt$, and a tree $q''$ that simulates $q'$.
To this end, we show something stronger, namely that $\pop$ \emph{commutes} with $\widehat{f}(q)$, in the following sense.

\begin{restatable}{lemma}{restatepreservepop}%
\label{lemma:preserve-pop}
Let $t_1, t_2 \in \Topo$, and let $f$ be an embedding of $t_1$ inside $t_2$.
Now $\pop$ is compatible with $\widehat{f}$, i.e., if $\pop(q) = (pkt, q')$, then $\pop(\widehat{f}(q)) = (pkt, \widehat{f}(q'))$.
\end{restatable}
\begin{proof}[Proof sketch]
By induction on $t_1$.
In the base, where $t_1 = *$, the claim holds because $\widehat{f}(q) = q$.
For the inductive step, let $t_1 = \mathsf{Node}(\vec{t}_1)$ and $n = |\vec{t}_1|$.
Furthermore, let $i$ be the index at the head of the PIFO attached to the root of $q$.
As before, we proceed by inverse induction on the prefixes $\addr$ of $f(j)$ with $1 \leq j \leq n$, arguing more generally that the following hold:
\begin{enumerate}
    \item
    If $\addr$ is a prefix of $f(i)$, then $\pop(\widehat{f}(q)_{\addr}) = (pkt, \widehat{f}(q')_{\addr})$.
    \item
    Otherwise, if $\addr$ is not a prefix of $f(i)$, then $\widehat{f}(q)_{\addr} = \widehat{f}(q')_{\addr}$.
\end{enumerate}
Instantiating the first property for the root address $\epsilon \in \mathsf{Path}(t_1)$ then tells us that:
\[
    \pop(\widehat{f}(q))
        = \pop(\widehat{f}(q)_\epsilon)
        = (pkt, \widehat{f}(q')_\epsilon)
        = (pkt, \widehat{f}(q')).
    \qedhere
\]
\end{proof}

\begin{remark}
For any $t \in \Topo$, we can define an operator that takes two PIFO trees $q_1, q_2 \in \mathsf{PIFOTree}(t)$ and produces a new PIFO tree $q_1 ; q_2 \in \mathsf{PIFOTree}(t)$, where the contents of the PIFO at each node in $q_2$ are concatenated to the corresponding node in $q_1$ (by shifting the weights of the former and inserting them).
This operator is associative, and has the empty PIFO tree as its neutral element on both sides, which makes $\mathsf{PIFOTree}(t)$ a monoid.
When restricted to well-formed PIFO trees, it is a \emph{free monoid}: every PIFO tree can uniquely be written as the concatenation of PIFO trees where each PIFO holds at most one element.
Further, given an embedding $f$ of $t_1$ into $t_2$, $\widehat{f}$ is a monoid homomorphism from $\mathsf{PIFOTree}(t_1)$ to $\mathsf{PIFOTree}(t_2)$.
These properties can then be exploited to obtain an abstract proof of \Cref{lemma:preserve-pop}.
Here we have chosen to give an explicit proof.
\end{remark}

\subsection{Preserving \texorpdfstring{$\push$}{push}}

We continue by showing that $\widehat{f}$ is also compatible with $\push$.
To this end, we must show that, if we $\push$ the packet $pkt$ into $q$ with path $pt$, then there exists a path $pt'$ such that $\push$ing $pkt$ into $\widehat{f}(q)$ with path $pt'$ results in a tree that simulates $\push(q, pkt, pt)$.
As it turns out, the correspondence between paths in $t_1$ and paths in $t_2$ can \emph{also} be obtained from $f$---intuitively, it extends paths that lead to a leaf of $t_1$ to the corresponding leaf in $t_2$, while duplicating the ranks as necessary.

For example, let us revisit the topologies $t_1$ and $t_2$ from \Cref{fig:example-embedding-a} and $\push$ a packet into a PIFO tree with topology $t_1$ using the path $(2, 5) :: 7$.
To preserve $\push$ on a tree with topology $t_2$, the path is translated into $(2, 5) :: (1, 5) :: 7$ by the embedding from the same figure.
See the lower row of \Cref{fig:lifting} for an illustration.
We give a formal definition of this translation below.

\begin{definition}
Let $t_1, t_2 \in \Topo$ and let $f$ be an embedding of $t_1$ inside $t_2$.
We define $\widetilde{f}: \mathsf{Path}(t_1) \to \mathsf{Path}(t_2)$ by induction on $t_1$.
First, if $t_1 = *$, then every $pt \in \mathsf{Path}(t_1)$ is of the form $r$ for $r \in \mathsf{Rnk}$; we set $\widetilde{f}(r) = r$.
Since $t_2 = *$ by \Cref{lemma:decompose-embedding}, this is well-defined.

Otherwise, if $t_1 = \mathsf{Node}(\vec{t}_1)$ with $|\vec{t}_1| = n$, then for $1 \leq i \leq n$ let $f_i$ be the embedding of $t_1/i$ into $t_2/f(i)$.
Every element of $\mathsf{Path}(t_1)$ is of the form ${(i,r) :: pt}$ where $1 \leq i \leq n$, $r \in \mathsf{Rk}$ and $pt \in \mathsf{Path}(t_1/i)$.
For every prefix $\addr$ of $f(i)$, we define ${\widetilde{f}((i, r) :: pt)_{\addr}}$ by inverse recursion.
In the base, where ${\addr = f(i)}$, we set ${\widetilde{f}((i, r) :: pt)_{\addr} = \widetilde{f_i}(pt)}$.
In the inductive step, where $\addr = \addr' \cdot j$, we set $\widetilde{f}((i, r) :: pt)_{\addr} = (j, r) :: \widetilde{f}((i, r) :: pt)_{\addr'}$.
Finally, we define $\widetilde{f}((i, r) :: pt) = \widetilde{f}((i, r) :: pt)_\epsilon$.
\end{definition}

We can now show that $\widehat{f}$ commutes with $\push$ if we translate the insertion path according to $\widetilde{f}$:

\begin{restatable}{lemma}{restatepreservepush}%
\label{lemma:preserve-push}
Let $t_1, t_2 \in \Topo$, $pkt \in \mathsf{Pkt}$, $pt \in \mathsf{Path}(t_1)$ and $q \in \mathsf{PIFOTree}(t_1)$.
We have
\[
    \widehat{f}(\push(q, pkt, pt)) = \push(\widehat{f}(q), pkt, \widetilde{f}(pt)).
\]
\end{restatable}
\begin{proof}[Proof sketch]
By induction on $t_1$.
In the base, where $t_1 = *$, we know ${\widehat{f}(q) = q}$ and ${\widetilde{f}(pt) = pt}$, and further, since ${\push(q, pkt, pt) \in \mathsf{PIFOTree}(t_1)}$, we know ${\widehat{f}(\push(q, pkt, pt)) = \push(q, pkt, pt)}$.
The claim then holds trivially.
For the inductive step, let $t_1 = \mathsf{Node}(\vec{t}_1)$ and $n = |\vec{t}_1|$.
We can then write $pt = (i, r) :: pt'$ where $pt' \in \mathsf{Path}(\vec{t}_1[i])$.
Let $\addr \in \mathsf{Addr}(t_1)$ be a prefix of some $f(j)$ with $1 \leq j \leq n$.
We can then prove by inverse induction that the following properties hold:
\begin{enumerate}
    \item
    If $\addr$ is a prefix of $f(i)$, then
    \(
        \widehat{f}(\push(q, pkt, pt))_{\addr}
            = \push(\widehat{f}(q)_{\addr}, pkt, \widetilde{f}(pt)_{\addr}).
    \)
    \item
    Otherwise, if $\addr$ is not a prefix of $f(i)$, then
    \(
        \widehat{f}(\push(q, pkt, pt))_{\addr} = \widehat{f}(q)_{\addr}
    \)
\end{enumerate}
When $\addr = \epsilon$, the first property instantiates to the claim.
\end{proof}

\subsection{Simulation}

Now that we know that lifted embeddings are compatible with the $\push$ and $\pop$ operations on PIFO trees, we can formally state and prove the correctness of this operation as follows.

\begin{theorem}%
\label{theorem:correctness}
Let $t_1, t_2 \in \Topo$, $q \in \mathsf{PIFOTree}(t_1)$.
If $f$ embeds $t_1$ in $t_2$, then $q \preceq \widehat{f}(q)$.
\end{theorem}
\begin{proof}
It suffices to prove that $R = \{ (q', f(q')) : q' \in \mathsf{PIFOTree}(t_1) \}$ is a simulation.
\begin{enumerate}
    \item
    If $\pop(q')$ is defined, then so is $\pop(\widehat{f}(q'))$ by \Cref{lemma:preserve-defined}.
    \item
    If $\pop(q') = (pkt, q'')$, then $\pop(\widehat{f}(q')) = (pkt, \widehat{f}(q''))$ by \Cref{lemma:preserve-pop}.
    \item
    If $\push(q', pkt, pt)$, then choose $pt' = \widetilde{f}(pt)$ to find by \Cref{lemma:preserve-push} that
    \[
        \push(q', pkt, pt)
            \; \mathrel{R} \; \widehat{f}(\push(q', pkt, pt))
            = \push(\widehat{f}(q'), pkt, pt').
        \qedhere
    \]
\end{enumerate}
\end{proof}

\noindent
Given an embedding $f$ of $t_1 \in \Topo$ into $t_2 \in \Topo$, we can now translate a control $c = (s, q, z) \in \mathsf{Control}(t_1)$ into a control $c' = (s, \widehat{f}(q), z') \in \mathsf{Control}(t_2)$, where $z'(s, pkt) = (\widetilde{f}(pt), s')$ when $z(s, pkt) = (pt, s')$.
\Cref{theorem:correctness} then tells us that $c'$ behaves just like $c$; there is no overhead in terms of state, and the translation of the scheduling transaction $z$ is straightforward.

\begin{remark}
Our results can be broadened to a more general model of scheduling.
Say $\pop$ returned a path containing all of the nodes whose PIFOs were popped, along with their ranks, and that a control could react to a $\pop$ using a ``descheduling transaction'' that could look at this $\pop$-path and update the state.
This extended version of $\pop$ would still be compatible with embedding, i.e., \Cref{lemma:preserve-pop} can be updated to show that if $\pop(q) = (pt, pkt, q')$, then $\pop(\widehat{f}(q)) = (\widetilde{f}(pt), pkt, \widehat{f}(q'))$.
\end{remark}

\subsection{A Counterexample}

A natural question to ask next is whether embeddings can be lifted in the opposite way, i.e., if $f$ embeds $t_1$ in $t_2$ and $q_2 \in \mathsf{PIFOTree}(t_2)$, can we use $f$ to find a $q_1 \in \mathsf{PIFOTree}(t_1)$ such that $q_2 \preceq q_1$?
\Cref{thm:simulations-vs-leaves} tells us that this is impossible if $t_2$ has more leaves than $t_1$.
But even if $t_1$ embeds in $t_2$ and $t_2$ does not have any more leaves than $t_1$, such a mapping may be impossible to find.

\begin{figure}[t]
\centering
\begin{tikzpicture}[->, >=stealth', auto, bend angle=10]
\tikzstyle{dot}=[state, inner sep=1.1pt, minimum size=0pt, draw=black, fill=black]
\tikzstyle{box}=[inner sep=0pt, minimum size=0pt, draw=black]
\small
  \node[box] (R) {\fbox{$\cdots{}3,2,1$}};
  \node[box] (C2) [below=16mm of R] {\fbox{$\cdots{}pkt_b$}};
  \node[box] (C1) [below left of=R, node distance=14mm] {\fbox{$\cdots{}pkt_a$}};
  \node[box] (C3) [below right of=R, node distance=14mm] {\fbox{$\cdots{}pkt_c$}};
  \path (R) edge node[swap, pos=.65, xshift=2pt, yshift=-2pt] {$\scriptstyle 1$} (C1);
  \path (R) edge node[swap, pos=.45, xshift=1pt] {$\scriptstyle 2$} (C2);
  \path (R) edge node[pos=.65, xshift=-2pt, yshift=-2pt] {$\scriptstyle 3$} (C3);
\end{tikzpicture}
\hspace{5mm}
\begin{tikzpicture}[->, >=stealth', auto, bend angle=10]
\tikzstyle{dot}=[state, inner sep=1.1pt, minimum size=0pt, draw=black, fill=black]
\tikzstyle{box}=[inner sep=0pt, minimum size=0pt, draw=black]
  \node[box] (epsA) {\fbox{$\cdots{}2,2,1$}};
  \node[box] (N0A) [below of=epsA, xshift=-10mm] {\fbox{$\cdots{}pkt_a$}};
  \node[box] (N1A) [below of=epsA, xshift=10mm] {\fbox{$\cdots{}2,1$}};
  \node[box] (N10A) [below of=N1A, xshift=-10mm] {\fbox{$\cdots{}pkt_b$}};
  \node[box] (N11A) [below of=N1A, xshift=10mm] {\fbox{$\cdots{}pkt_c$}};
  \path (epsA) edge node[swap, pos=.65, xshift=2pt, yshift=-2pt] {$\scriptstyle 1$} (N0A);
  \path (epsA) edge node[xshift=-2pt, pos=.65, yshift=-2pt] {$\scriptstyle 2$} (N1A);
  \path (N1A) edge node[swap, pos=.65, xshift=2pt, yshift=-2pt] {$\scriptstyle 1$} (N10A);
  \path (N1A) edge node[pos=.65, xshift=-2pt, yshift=-2pt] {$\scriptstyle 2$} (N11A);
\end{tikzpicture}
\hspace{5mm}
\begin{tikzpicture}[->, >=stealth', auto, bend angle=10]
\tikzstyle{dot}=[state, inner sep=1.1pt, minimum size=0pt, draw=black, fill=black]
\tikzstyle{box}=[inner sep=0pt, minimum size=0pt, draw=black]
\small
  \node[box] (epsA) {\fbox{$\cdots{}2,2,1,{\bf\color{red} 2}$}};
  \node[box] (N0A) [below of=epsA, xshift=-10mm] {\fbox{$\cdots{}pkt_a$}};
  \node[box] (N1A) [below of=epsA, xshift=10mm] {\fbox{$\cdots{\bf\color{red} 2},2,1$}};
  \node[box] (N10A) [below of=N1A, xshift=-10mm] {\fbox{$\cdots{}pkt_b$}};
  \node[box] (N11A) [below of=N1A, xshift=10mm] {\fbox{$\cdots{}{\color{red}pkt_d},pkt_c$}};
  \path (epsA) edge node[swap, pos=.65, xshift=2pt, yshift=-2pt] {$\scriptstyle 1$} (N0A);
  \path (epsA) edge node[xshift=-2pt, pos=.65, yshift=-2pt] {$\scriptstyle 2$} (N1A);
  \path (N1A) edge node[swap, pos=.65, xshift=2pt, yshift=-2pt] {$\scriptstyle 1$} (N10A);
  \path (N1A) edge node[pos=.65, xshift=-2pt, yshift=-2pt] {$\scriptstyle 2$} (N11A);
\end{tikzpicture}
\caption{Impossibility of simulation. The trees from left to right are $q_1', q_2', q_2''$.
Although $q_1'$ ostensibly simulates~$q_2'$, we can $\push$ $pkt_d$ into $q_2'$ to get $q_2''$, and no $\push$ of $pkt_d$ into $q_1'$ can simulate this.
}%
\label{fig:counterexample}
\end{figure}
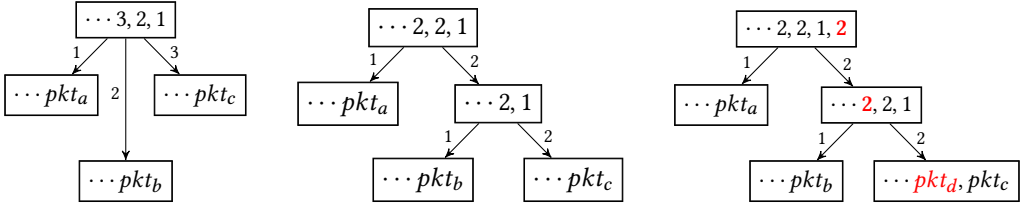

\begin{proposition}
There exist $t_1, t_2 \in \Topo$ s.t.\ for all PIFO trees $q_1 \in \mathsf{PIFOTree}(t_1)$ and $q_2 \in \mathsf{PIFOTree}(t_2)$, $q_2$ does not simulate $q_1$, even though $t_1$ embeds inside $t_2$ and $|t_1| = |t_2|$.
\end{proposition}
\begin{proof}
Consider the embedding shown in \Cref{fig:example-embedding-a}, and say that $t_1$ and $t_2$ are as labeled in that figure.
Consider any $q_1 \in \mathsf{PIFOTree}(t_1)$ and $q_2 \in \mathsf{PIFOTree}(t_2)$, and suppose towards a contradiction that $q_1$ simulates $q_2$.
Take three distinct packets $pkt_a$, $pkt_b$, and $pkt_c$, and let $q_2'$ be the tree obtained from $q_2$ by pushing $pkt_a$, $pkt_b$, and $pkt_c$ to the front of the first, second and third leaf respectively, such that $pkt_a$, $pkt_b$, and $pkt_c$ will be popped in that order, as depicted in the center in~\Cref{fig:counterexample}.

Because $q_1$ simulates $q_2$, we can $\push$ these packets to obtain a tree $q_1'$ that simulates $q_2'$.
In this tree, $pkt_a$, $pkt_b$, and $pkt_c$ \emph{must} appear in different leaves---otherwise, we could execute a series of $\push$es and $\pop$s (as in \Cref{lemma:permutation}) that changed the order of two packets in $q_2'$ but not in $q_1'$.
Furthermore, because these packets are popped first from $q_2'$, they are also popped first from $q_1'$.
We can therefore assume (without loss of generality) that $q_1'$ is as depicted on the left of \Cref{fig:counterexample}.

Now take a fourth packet $pkt_d$ distinct from $pkt_a$, $pkt_b$ and $pkt_c$, and $\push$ it into $q_2'$ to obtain the PIFO tree $q_2''$, depicted on the right in \Cref{fig:counterexample} (new elements in {\color{red}red}).
In $q_2''$, the first four packets $\pop$ped are $pkt_b$, $pkt_a$, $pkt_c$, and $pkt_d$.
Because $q_1'$ simulates $q_2'$, there exists a way to $\push$ $pkt_d$ into $q_1'$ that results in a PIFO tree $q_1''$ where the first four packets $\pop$ped match those of $q_2''$.
We can constrain the position of $pkt_d$ in $q_1''$: $pkt_d$ \emph{must} be enqueued at the second position of the third leaf of $q_1'$---otherwise, if, for instance, $pkt_d$ appeared in the same leaf as $pkt_a$ in $q_1''$, then $q_2''$ could swap the relative order of $pkt_d$ and $pkt_a$, but $q_1''$ could not (via the same technique as in \Cref{lemma:permutation}).

Together, this means that, in $q_1''$, {\color{red}$3$} must be enqueued near the head of the root PIFO\@.
We have four possibilities for the root PIFO of $q_1''$, listed here along with the first four pops they would effect:
\begin{mathpar}
    321{\color{red}3}: pkt_c, pkt_a, pkt_b, pkt_d
    \and
    32{\color{red}3}1: pkt_a, pkt_c, pkt_b, pkt_d
    \and
    3{\color{red}3}21/{\color{red}3}321: pkt_a, pkt_b, pkt_c, pkt_d
\end{mathpar}
Because none of these match the first four packets popped from $q_2''$, we have reached a contradiction.
Our assumption that $q_1$ simulates $q_2$ must therefore be false, which proves the claim.
\end{proof}

\section{Embedding Algorithms}%
\label{sec:embedding_algs}

So far, we showed that embeddings can be used to replicate the behavior of one PIFO tree with a PIFO tree of a different topology.
To exploit these results in practice, we need to calculate an embedding of one (user-supplied) topology into another (hardware-mandated) topology.
This section proposes two efficient algorithms that can be used to find such an embedding, if it exists.

\subsection{Embedding into Complete \texorpdfstring{$d$}{d}-ary Topologies}

We start by treating the special case where the target of our embedding is a complete topology of fixed arity.
In this case, the algorithm also has a favorable complexity.

\begin{theorem}%
\label{thm:embed1}
Let $t_1, t_2 \in \Topo$ such that $t_2$ is a complete $d$-ary topology.
There is an $O(n\log n)$ algorithm ($n = |\mathsf{Addr}(t_1)|$) to determine whether $t_1$ embeds in $t_2$, and to find such an embedding if so.
\end{theorem}
\begin{proof}
We construct an embedding $f$ of $t_1$ into complete $d$-ary topology of minimal height, in a greedy bottom-up fashion reminiscent of the construction of optimal Huffman codes~\cite{CoverThomas06}.
Let $\addr \in \mathsf{Addr}(t_1)$ and suppose, at some stage of the algorithm, we have for each $\addr \cdot i \in \mathsf{Addr}(t_1)$ an embedding of $t_1/(\addr \cdot i)$ inside a complete $d$-ary topology $t^d_i$ of minimum height---in particular, when $\addr \cdot i$ is a leaf, this is the trivial embedding.
To get a minimum-height embedding of $t_1/\addr$ inside a complete $d$-ary topology, we insert all of these topologies $t^d_i$ into a min-priority queue, ranked by height.
We repeat the following until the queue has one element:
\begin{enumerate}
    \item Extract up to $d$ elements of the same minimum priority (say $m$).
    \item Create a new topology of height $m+1$ with those elements as children.
    If desired, pad this new topology---add dummy leaves just below the root---to make it a \emph{complete} $d$-ary topology.
    \item Insert this new topology with height $m+1$ into the priority queue with priority $m+1$.
\end{enumerate}
Note, if there is only one topology of minimum height $m$, we do not need to form a new topology with one child.
We simply reinsert the topology, unchanged, into the queue but with priority $m+1$.

It follows from an exchange argument (given below; see~\cite[Ch.~4]{KleinbergTardo05} and~\cite{dangelo}) that the single remaining topology $t$ is of the minimum height into which $t_1/\addr$ embeds.
An embedding of the whole topology is thus possible provided that $t$ is no taller than $t_2$.
The time required to process a node with $d_s$ children is $O(d_s \log d_s)$ for the maintenance of the priority queue and $O(d_s)$ for all other operations, or $O(n\log n)$ in all, where $n = |t_1|$.

Here is the exchange argument.
At any stage of the algorithm, let $t_1', \dots, t_k'$ be a maximal set, up to size $d$, of queue elements with minimum priority $m$.
Suppose the action of forming a new topology of height $m+1$ with these children and inserting it into the queue was not the right thing to do, i.e., some other action would have led to a topology $t'$ of optimal height.
Assume without loss of generality that $t_1'$ occurs deepest in $t'$; thus the path from the root to $t_1'$ is the longest among the queue elements $t_i'$.
Permute $t'$ so that $t_1'$ occurs rightmost.
The topologies $t_2', \dots, t_k'$ can be swapped for the siblings of $t_1'$ in $t'$, or, if $t_1'$ has fewer than $k$ siblings, added as a new child to the parent of $t_1'$ in $t'$, without increasing the height of the topology.
Thus the action of the algorithm was correct.

For the case $k=1$ with more than one topology left in the queue, the argument is the same, with the observation that in any embedding, $t_1'$ must have a sibling of greater height.
Therefore, it does not hurt to consider $t_1'$ to be of height $m+1$ even though its height is actually $m$.
\end{proof}

\subsection{Embedding into Arbitrary Topologies}

We now turn our attention to the more general case, where the target of embedding is an arbitrary PIFO topology.
An embedding algorithm can still be achieved here, at the cost of a higher complexity.

\begin{theorem}%
\label{thm:embed2}
Let $t_1, t_2 \in \Topo$, such that each node in $t_1$ has at most $d$ children.
There is a polynomial-time algorithm to determine whether $t_1$ embeds in $t_2$, and to find such an embedding if so.
\end{theorem}
\begin{proof}
The algorithm is based on bottom-up dynamic programming.
Construct a boolean-valued table~$T$ with entries $T(\addr_1,\addr_2)$ for $\addr_1 \in \mathsf{Addr}(t_1)$ and $\addr_2 \in \mathsf{Addr}(t_2)$ that says whether there exists an embedding that maps $\addr_1$ to $\addr_2$.
Suppose that $\addr_1 \in \mathsf{Addr}(t_1)$ has $d_1$ children, that we have determined the values of $T(\addr_1', \addr_2')$ such that $\addr_1$ (resp. $\addr_2$) is an ancestor (i.e., prefix) of $\addr_1'$ (resp. $\addr_2'$), and that we wish to determine $T(\addr_1,\addr_2)$.
The value will be $\mathsf{true}$ if we can find target nodes for the children of $\addr_1$ among the descendants of $\addr_2$ that form an antichain w.r.t.\ the ancestor/descendant relation.

Without loss of generality, we can ignore node $\addr_2'$ as a candidate target for a child $\addr_1 \cdot i$ if $\addr_2'$ has a descendant $\addr_2''$ such that $T(\addr_1 \cdot i,\addr_2'')=\mathsf{true}$; any embedding that maps $\addr_1 \cdot i$ to $\addr_2'$ can map $\addr_1 \cdot i$ to $\addr_2''$ instead.
We now form a graph with nodes $(\addr_1',\addr_2')$ such that $T(\addr_1',\addr_2')=\mathsf{true}$ and edges $((\addr_1',\addr_2'),(\addr_1'',\addr_2''))$ such that either $\addr_1'=\addr_1''$, or $\addr_2'=\addr_2''$, or $\addr_2'$ is an ancestor of $\addr_2''$.
An independent set of size $d_1$ in this graph consists of pairs $(\addr_1',\addr_2')$ such that the first components are all the $d_1$ children of $\addr_1$ and the second components form an antichain among the descendants of $\addr_2$ to which the first components can map.
If such an independent set exists, then we set $T(\addr_1,\addr_2)$ to $\mathsf{true}$.

Let $n=|\mathsf{Addr}(t_1)|$ and $m=|\mathsf{Addr}(t_2)|$.
To calculate $T(\addr_1, \addr_2)$ for each of the $nm$ pairs $\addr_1, \addr_2$, we need to find an independent set in a graph with $dm$ nodes, which can be solved with brute force in time $O(m^d)$, and this dominates the complexity.
The total time of the algorithm is $O(nm^{d+1})$.
\end{proof}

\section{Implementation}%
\label{sec:implementation}

\definecolor{red}{rgb}{1.0, 0.0, 0.0}
\definecolor{skyblue}{rgb}{0.53, 0.81, 0.92}
\definecolor{forestgreen}{rgb}{0.13, 0.55, 0.13}
\definecolor{lightsalmon}{rgb}{1.0, 0.63, 0.48}
\definecolor{dodgerblue}{rgb}{0.12, 0.56, 1.0}
\definecolor{darkseagreen}{rgb}{0.56, 0.74, 0.56}
\definecolor{orchid}{rgb}{0.85, 0.44, 0.84}

\newcommand\cfbox[2]{\fcolorbox{black}{#1}{#2}}

\newcommand\abox{\cfbox{red}{A}}
\newcommand\bbox{\cfbox{skyblue}{B}}
\newcommand\cbox{\cfbox{forestgreen}{C}}
\newcommand\dbox{\cfbox{lightsalmon}{D}}
\newcommand\ebox{\cfbox{dodgerblue}{E}}
\newcommand\ffbox{\cfbox{darkseagreen}{F}}
\newcommand\gbox{\cfbox{orchid}{G}}

We implemented PIFO trees as described in~\S\ref{sec:overview}, along with the embedding algorithm covered in~\Cref{thm:embed1}, in OCaml.
Our implementation includes a simulator that can ``run'' a traffic sample of packets through a PIFO tree and visualize the results with respect to $\push$ and $\pop$ time.

We show that our scheduler gives expected results when programming standard scheduling algorithms.
We compare the traces from our scheduler to those obtained from a state-of-the-art hardware switch and find reasonable correlation.
This reinforces the idea that PIFO trees are a reasonable primitive for standard programmable scheduling.
We then move beyond the capabilities of a modern switch by showing the result of running hierarchical algorithms on our scheduler.
Finally, we show how our implementation supports the embedding of one PIFO tree into another without having to program the latter tree anew and without appreciable loss in performance.

\subsection{Preliminaries}%
\label{sec:impl_preliminaries}

A PIFO tree is programmed, as in the presentation thus far, by specifying a topology and a control.
The simulator takes a collection of incoming packets, represented as a \pcap, and attempts to $\push$ the packets into the tree.
It paces these $\push$es to match the packets' timestamps in the given \pcap.
The simulator also requires a \emph{line rate}, which is the frequency at which it automatically calls $\pop$.

For ease of presentation, we have standardized a few things that are not fundamental to our implementation.
All our traffic samples contain $60$ packets of the same size, coming from seven source addresses that we label A through G.
We partition the traffic into flows simply based on these addresses, though more sophisticated partitioning is of course possible.

An important subtlety is \emph{saturation}: we need to enqueue packets quickly enough that the switch actually has some backlog of packets and therefore has to make a scheduling decision among those packets.
This is easy to do in our OCaml simulator, where we have total control: we set packets to arrive at the rate of $10$ packets/s and simulate a slower line rate of $4$ packets/s.
We also achieve similar saturation when running comparative experiments on the hardware switch.

\paragraph{Reading our visualizations}
In the visualizations that follow, the x-axis, read left to right, shows the time (in seconds) since the simulation started; the y-axis, read top to bottom, shows packets in the order they arrive at the scheduler.
A colored horizontal bar represents a packet, and is colored based on its source address: \abox\bbox\cbox\dbox\ebox\ffbox\gbox.
A horizontal bar is drawn starting at the time when the packet arrives at the scheduler and ending at the time when it is released by the scheduler, so shorter bars indicate preferential service.
Later packets tend to have longer bars precisely because we are saturating the scheduler.
To notice trends generally, it is useful to zoom out and look at the ``shadow'' that a certain-colored flow casts.
To focus more specifically on which packets are released when, it is useful to zoom in and study the right edges of the horizontal bars.

\subsection{Standard Algorithms}

\begin{table}\sffamily
  \caption{Standard scheduling algorithms as run against our scheduler and a programmable hardware switch.}%
  \begin{tabular}{l*2{C}@{}}
  \toprule
   & Through a PIFO tree & Through a programmable switch \\
  \midrule
  \vertlabel{FCFS} & \pic[fcfs] & \pic[tofino_fcfs]  \\
  \vertlabel{Strict} & \pic[strict] & \pic[tofino_strict] \\
  \vertlabel{R'Robin} & \pic[rr] & \pic[tofino_rr] \\
  \vertlabel{WFQ} & \pic[wfq] & \pic[tofino_wfq] \\
  \bottomrule
  \end{tabular}%
  \label{table:standard}
  \end{table}

We first show that PIFO trees are a reasonable primitive for standard, non-hierarchical algorithms.
We use a simple ternary topology: one root node with three leaves below it.
The standard algorithms discussed below correspond to the rows of \Cref{table:standard}.
We have constructed a \pcap~of packets with sources \abox\bbox\cbox, and, in the left column, we schedule that \pcap~using PIFO trees.

The first come, first served (FCFS) algorithm transmits packets in the order they are received.
The strict algorithm strictly prefers C to B and B to A, up to availability of input packets.
It is easy to see the expected trend in the visualization: focus on the ``shadow'' casy by each flow; A's is the smallest, followed by B's, followed by C's.
Round-robin seeks to alternate between flows A, B, and C.
Weighted fair queueing (WFQ)~\cite{demers} allows user-defined weights for each flow; in the table we show a split of 10/20/30 between A/B/C.
The shadows guide intuition here as well.

As a further, albeit informal, endorsement of PIFO trees, we compare our tree-based scheduling trends to those of a modern switch.
The right column of \Cref{table:standard} shows the same algorithms and \pcap s scheduled by the FIFO-based programmable scheduler exposed by our hardware switch.
We find the trends to be in agreement.
Our testbed consists of two host Linux machines connected by a hardware switch.
To keep the total number of packets small enough to present here, we use a slow ingress rate of $60$ kbps and an egress rate of $24$ kbps, with packets of size $1$ kb.
Although it is untenable to visualize them here, we also conduct similar experiments with more realistic ingress and egress rates ($2.5$ Mbps and $1$ Mbps respectively) and observe similar scheduling trends.

We point out a few oddities in the hardware experiments:
\begin{enumerate}
\item
The first few packets going into the switch spend so little time in the queue that the horizontal bars representing them do not render in our visualizations.
\item
The switch releases packets not one by one but in batches of four.
This trend does not appear in our more realistic experiments. We conjecture that this stagger is an artifact of our unrealistically slow rates: the egress thread appears to pull packets in batches.
\item
The packets do not arrive in as regular a pattern as in software; see the left edges of the horizontal bars.
This is perhaps to be expected in a real-world experiment.
\end{enumerate}

\subsection{Hierarchical Algorithms}

\begin{table}\sffamily
  \caption{HPFQ running a traffic flow. Share ratios as indicated.}%
  \begin{tabular}{l*2{C}@{}}
  \toprule
   & Tree sketch & Visualization \\
  \midrule
  \vertlabel{HPFQ} & \begin{tikzpicture}[->, >=stealth', auto, bend angle=10]
\tikzstyle{dot}=[state, inner sep=1.3pt, minimum size=0pt, draw=black, fill=black]
\tikzstyle{box}=[inner sep=0pt, minimum size=0pt, draw=black]
\small
\node[box] (epsA) [node distance=5cm] {\fbox{WFQ: 80/20}};
\node[box] (N0A) [below of=epsA, xshift=8mm] {\cbox};
\node[box] (N1A) [below of=epsA, xshift=-8mm] {\fbox{WFQ: 75/25}};
\node[box] (N10A) [below of=N1A, xshift=-8mm] {\abox};
\node[box] (N11A) [below of=N1A, xshift=8mm] {\bbox};
\path (epsA) edge node[swap, pos=.65, xshift=2pt, yshift=4.5pt] {} (N0A);
\path (epsA) edge node[xshift=-2pt, pos=.65, yshift=4.5pt] {} (N1A);
\path (N1A) edge node[swap, pos=.65, xshift=2.5pt, yshift=-2pt] {} (N10A);
\path (N1A) edge node[pos=.65, xshift=-2pt, yshift=-2pt] {} (N11A);
\end{tikzpicture} & \pic[hpfq] \\
  \bottomrule
  \end{tabular}%
  \label{table:hpfq}
\end{table}

\begin{table}\sffamily
  \caption{Hierarchical algorithms (rows 1 and 3), and compiling ternary trees to binary (rows 2 and 4).}%
  \begin{tabular}{l*2{C}@{}}
  \toprule
   & Tree sketch & Visualization \\
  \midrule
  \vertlabel{TwoPol} &  \begin{tikzpicture}[->, >=stealth', auto, bend angle=10]
 \tikzstyle{dot}=[state, inner sep=1.3pt, minimum size=0pt, draw=black, fill=black]
 \tikzstyle{box}=[inner sep=0pt, minimum size=0pt, draw=black]
 \small
 \node[box] (root) [node distance=5cm] {\fbox{WFQ: 10/10/80}};
 \node[box] (N1) [below of=root, xshift=-8mm] {\abox};
 \node[box] (N2) [below of=root, xshift=0mm] {\bbox};
 \node[box] (root2) [below of=root, xshift=18mm] {\fbox{Strict: E$>$D$>$C}};
 \node[box] (N1') [below of=root2, xshift=-8mm] {\cbox};
 \node[box] (N2') [below of=root2, xshift=0mm] {\dbox};
 \node[box] (N3') [below of=root2, xshift=8mm] {\ebox};
 \path (root) edge node[swap, pos=.65, xshift=2pt, yshift=4.5pt] {} (N1);
 \path (root) edge node[xshift=-2pt, pos=.65, yshift=4.5pt] {} (N2);
 \path (root) edge node[swap, pos=.65, xshift=2.5pt, yshift=-2pt] {} (root2);
 \path (root2) edge node[swap, pos=.65, xshift=2pt, yshift=4.5pt] {} (N1');
 \path (root2) edge node[xshift=-2pt, pos=.65, yshift=4.5pt] {} (N2');
 \path (root2) edge node[swap, pos=.65, xshift=2.5pt, yshift=-2pt] {} (N3');
 \end{tikzpicture} & \pic[fairstrict2tier] \\
  \vertlabel{TwoPol Bin} &  \begin{tikzpicture}[->, >=stealth', auto, bend angle=10]
  \tikzstyle{dot}=[state, inner sep=1.3pt, minimum size=0pt, draw=black, fill=black]
  \tikzstyle{box}=[inner sep=0pt, minimum size=0pt, draw=black]
  \small
  \node[box] (root) [node distance=5cm] {\fbox{WFQ: 10/10/80}};
    \node[box, fill=gray] (N2) [below of=root, xshift=-8mm] {\fbox{T}};
      \node[box] (N21) [below of=N2, xshift=-4mm] {\abox};
      \node[box] (N22) [below of=N2, xshift=4mm] {\bbox};
    \node[box] (N1) [below of=root, xshift=9mm] {\fbox{{Strict: E$>$D$>$C}}};
      \node[box, fill=gray] (N11) [below of=N1, xshift=-4mm] {\fbox{T}};
        \node[box] (N111) [below of=N11, xshift=-4mm] {\cbox};
        \node[box] (N112) [below of=N11, xshift=4mm] {\dbox};
      \node[box] (N12) [below of=N1, xshift=4mm] {\ebox};
  \path (root) edge node[swap, pos=.65, xshift=2pt, yshift=4.5pt] {} (N1);
  \path (root) edge node[xshift=-2pt, pos=.65, yshift=4.5pt] {} (N2);
  \path (N1) edge node[swap, pos=.65, xshift=2pt, yshift=4.5pt] {} (N11);
  \path (N1) edge node[xshift=-2pt, pos=.65, yshift=4.5pt] {} (N12);
  \path (N2) edge node[swap, pos=.65, xshift=2pt, yshift=4.5pt] {} (N21);
  \path (N2) edge node[xshift=-2pt, pos=.65, yshift=4.5pt] {} (N22);
  \path (N11) edge node[swap, pos=.65, xshift=2pt, yshift=4.5pt] {} (N111);
  \path (N11) edge node[xshift=-2pt, pos=.65, yshift=4.5pt] {} (N112);
\end{tikzpicture} & \pic[fairstrict2tier_bin] \\
  \vertlabel{3Tier3} &  \begin{tikzpicture}[->, >=stealth', auto, bend angle=10]
 \tikzstyle{dot}=[state, inner sep=1.3pt, minimum size=0pt, draw=black, fill=black]
 \tikzstyle{box}=[inner sep=0pt, minimum size=0pt, draw=black]
 \small
 \node[box] (root) [node distance=5cm] {\fbox{WFQ: 40/40/20}};
 \node[box] (N1) [below of=root, xshift=-8mm] {\abox};
 \node[box] (N2) [below of=root, xshift=0mm] {\bbox};
 \node[box] (root2) [below of=root, xshift=8mm] {\fbox{RR}};
 \node[box] (N1') [below of=root2, xshift=-8mm] {\cbox};
 \node[box] (N2') [below of=root2, xshift=0mm] {\dbox};
 \node[box] (root3) [below of=root2, xshift=16mm] {\fbox{WFQ: 10/40/50}};
 \node[box] (N1'') [below of=root3, xshift=-8mm] {\ebox};
 \node[box] (N2'') [below of=root3, xshift=0mm] {\ffbox};
 \node[box] (N3) [below of=root3, xshift=8mm] {\gbox};
 \path (root) edge node[swap, pos=.65, xshift=2pt, yshift=4.5pt] {} (N1);
 \path (root) edge node[xshift=-2pt, pos=.65, yshift=4.5pt] {} (N2);
 \path (root) edge node[swap, pos=.65, xshift=2.5pt, yshift=-2pt] {} (root2);
 \path (root2) edge node[swap, pos=.65, xshift=2pt, yshift=4.5pt] {} (N1');
 \path (root2) edge node[xshift=-2pt, pos=.65, yshift=4.5pt] {} (N2');
 \path (root2) edge node[swap, pos=.65, xshift=2.5pt, yshift=-2pt] {} (root3);
 \path (root3) edge node[swap, pos=.65, xshift=2pt, yshift=4.5pt] {} (N1'');
 \path (root3) edge node[xshift=-2pt, pos=.65, yshift=4.5pt] {} (N2'');
 \path (root3) edge node[swap, pos=.65, xshift=2.5pt, yshift=-2pt] {} (N3);
 \end{tikzpicture} & \pic[fair3tier] \\
  \vertlabel{3Tier3 Bin} &  \begin{tikzpicture}[->, >=stealth', auto, bend angle=10]
  \tikzstyle{dot}=[state, inner sep=1.3pt, minimum size=0pt, draw=black, fill=black]
  \tikzstyle{box}=[inner sep=0pt, minimum size=0pt, draw=black]
  \small
  \node[box] (root) [node distance=5cm] {\fbox{WFQ: 40/40/20}};
    \node[box] (N1) [below of=root, xshift=12mm, yshift=4mm] {\fbox{RR}};
      \node[box] (N11) [below of=N1, xshift=6mm] {\fbox{WFQ: 10/40/50}};
        \node[box, fill=gray] (N111) [below of=N11, xshift=-4mm] {\fbox{T}};
          \node[box] (N1111) [below of=N111, xshift=-4mm] {\ebox};
          \node[box] (N1112) [below of=N111, xshift=4mm] {\ffbox};
        \node[box] (N112) [below of=N11, xshift=4mm] {\gbox};
      \node[box, fill=gray] (N12) [below of=N1, xshift=-8mm] {\fbox{T}};
        \node[box] (N121) [below of=N12, xshift=-8mm] {\cbox};
        \node[box] (N122) [below of=N12, xshift=0mm] {\dbox};

    \node[box, fill=gray] (N2) [below of=root, xshift=-12mm, yshift=4mm] {\fbox{T}};
      \node[box] (N21) [below of=N2, xshift=-4mm] {\abox};
      \node[box] (N22) [below of=N2, xshift=4mm] {\bbox};

  \path (root) edge node[swap, pos=.65, xshift=2pt, yshift=4.5pt] {} (N1);
  \path (root) edge node[xshift=-2pt, pos=.65, yshift=4.5pt] {} (N2);
  \path (N1) edge node[swap, pos=.65, xshift=2pt, yshift=4.5pt] {} (N11);
  \path (N1) edge node[xshift=-2pt, pos=.65, yshift=4.5pt] {} (N12);
  \path (N11) edge node[swap, pos=.65, xshift=2pt, yshift=4.5pt] {} (N111);
  \path (N11) edge node[xshift=-2pt, pos=.65, yshift=4.5pt] {} (N112);
  \path (N111) edge node[swap, pos=.65, xshift=2pt, yshift=4.5pt] {} (N1111);
  \path (N111) edge node[xshift=-2pt, pos=.65, yshift=4.5pt] {} (N1112);
  \path (N12) edge node[swap, pos=.65, xshift=2pt, yshift=4.5pt] {} (N121);
  \path (N12) edge node[xshift=-2pt, pos=.65, yshift=4.5pt] {} (N122);
  \path (N2) edge node[swap, pos=.65, xshift=2pt, yshift=4.5pt] {} (N21);
  \path (N2) edge node[xshift=-2pt, pos=.65, yshift=4.5pt] {} (N22);

\end{tikzpicture} & \pic[fair3tier_bin] \\
  \bottomrule
  \end{tabular}%
  \label{table:compilation}
\end{table}

Our implementation supports arbitrary hierarchical PIFO trees; there is no equivalent in our programmable hardware switch.
Hierarchical packet fair queueing (HPFQ)~\cite{bennett} is WFQ applied at many levels, and is a more general instance of the motivating problem we discussed in~\S\ref{sec:motivating_trees}.
\Cref{table:hpfq} shows HPFQ scheduling a traffic sample, now with unequal splits as shown.
Paying attention to the right hand side of the horizontal colored bars, we see that B gets excellent service (80/20, versus C) until the first packets from A arrive, after which B's share decreases.
This could not have been recreated using fair scheduling without hierarchies.
Additionally, trees need not run the same algorithm on different nodes.
The first and third rows of \Cref{table:compilation} show examples of more complicated trees running combinations of algorithms on their nodes.

\subsection{Compilation}

We implement a compilation from ternary PIFO trees to binary PIFO trees using the technique described in~\Cref{thm:embed1}, meaning that we can take any ternary topology along with a control written against it, and automatically create a binary topology and control that together simulate the former.
We use this to compile all the ternary algorithms described in \Cref{table:standard} into binary PIFO trees, and we find the resultant visualizations unchanged.
We also compile the complex trees shown in rows 1 and 3 of \Cref{table:compilation} into binary trees; these are shown in rows 2 and 4.
The transit nodes that were automatically created are shown in gray, and the visualizations produced remain unchanged.

Our visualizations are sensitive to $\push$ and $\pop$ time, and so it is encouraging to see them stay the same across the compilation process: this suggests that we experience no appreciable loss in performance as a result of the compilation despite the introduction of new intermediary nodes.

\section{Related Work}

We outline a class of algorithms that we have not considered in our formalization.
We review other efforts in formalizing schedulers.
We study other efforts in programmable packet scheduling.

\subsection{Non Work-Conserving Algorithms}%
\label{sec:shaping}

The focus of our study has been \emph{work-conserving} scheduling: a $\push$ed packet can immediately be $\pop$ped.
\emph{Non work-conserving} algorithms, also known as shaping algorithms, say that a packet cannot be $\pop$ped until some \emph{time} that is computed, specifically for that packet, at $\push$.
This means that a less favorably ranked packet may be released before a more favorably ranked one if the former is ready and the latter is not, and the link may go idle if no packets are ready.
We leave shaping for future work, but provide a few pointers.
\citet{Sivaraman16} include shaping in their PIFO tree model.
Loom~\cite{loom} repeatedly reinserts shaped packets into the tree.
Carousel~\cite{carousel} takes a more general approach, applying shaping to all algorithms.

\subsection{Formalizations in the Domain of Packet Scheduling}

There is a wealth of work from the algorithms and theory communities towards formally studying packet scheduling using competitive analysis~\cite{aiello, mansour, kesselman}.
We refer interested readers to a comprehensive survey by \citet{goldwasser}.

\citet{chakareski2011} formalizes in-network packet scheduling for video transmission as a constrained optimization problem expressed using Lagrange multipliers.
Nodes coordinate to compute the optimal rate at which other nodes should send packets.
\citet{durr2016} map packet scheduling to the no-wait job-shop scheduling problem from operations research, arriving at an integer linear program that exposes the constraints under which to minimize maximum congestion.

SP-PIFOs~\cite{Alcoz20} orchestrate a collection of FIFOs to approximate the behavior of a PIFO\@.
This is bound to be imperfect, so they give a formal way to measure the number of mistakes their model makes relative to a perfect PIFO\@.
Follow-on analysis by \citet{vass2022} problematizes Alcoz et al.'s push-up/push-down heuristic, showing that it can introduce mistakes linearly up to the number of FIFOs. %
Their solution is a new heuristic called Spring that counts packets using exponentially weighted moving averages and can achieve twofold speedup compared to SP-PIFOs.

In search of a universal packet scheduler, \citet{mittal2015} develop a general model of packet scheduling.
They formally define a \emph{schedule} as a set of the (fixed set of) packets, the paths those packets take through the network, and the input and output times the packets enjoy.
A schedule is effected by the set of scheduling algorithms that the routers along the way implement.
One set of scheduling algorithms can \emph{replay} another if, for all packets in the fixed set, it gives each packet an equal or better output time.
Mittal et al.\ show that there is no universal set of algorithms that can replay all others, but that the classic least slack time first algorithm~\cite{lstf} comes close.

\subsection{Programmable Scheduling}

We outline other lines of work that also orchestrate FIFOs and PIFOs to eke out more expressivity.
We also review work that allows higher-level programming of schedulers.

\paragraph{Orchestrating FIFOs and PIFOs}

SP-PIFOs~\cite{Alcoz20}, orchestrate a collection of FIFOs to approximate the behavior of a PIFO\@.
\citet{sharma2018} approximate WFQ on reconfigurable switches.
They present a scheduler called rotating strict priority, which transmits packets from multiple queues in approximate sorted order by grouping FIFOs into two groups and intermittently switching the relative priority of one group versus the other.
Tonic~\cite{tonic} can effect a range of scheduling algorithms using a FIFO along with the added notion of \emph{credit} to implement the round-robin algorithm.
\citet{pipo} study packet scheduling using the push-in pick-out (PIPO) data structure.
A PIPO is an approximation of the push-in extract-out (PIEO) queue~\cite{pieo} which picks out the best-ranked item that also satisfies some $\pop$-time predicate.
\citet{Sivaraman16} orchestrate PIFOs into a tree structure.
They provide a hardware design targeting shared memory switches.
They explain how a PIFO tree can be laid out in memory as a complete mesh, and sketch a compiler from the logical tree to the theoretical hardware mesh.

\paragraph{Programming at a higher level}

QtKAT~\cite{qtkat} extends NetKAT~\cite{netkat} to bring quality of service into consideration.
Their work paves the way for formal analysis using network calculus and, eventually, the verification of network-wide queueing.
\citet{snap} present SNAP, a tool that allows a ``one big switch'' abstraction: users can reason at a global level and program a fictitious big switch, and the tool checks the program for correctness and compiles it to the distributed setting.
NUMFabric~\cite{numfabric} allows the user to specify a utility maximization problem, once, and it then calculates the allocations of bandwidth that would maximize the utility function across the distributed network.
\citet{domino} allow data-plane programming as a high level ``packet transaction'' written in Domino, a new imperative language that compiles to line-rate hardware-level code running on programmable switches.

\section{Conclusion and Future Work}

This paper explored higher-level abstractions for packet scheduling.
Starting with a proposal by \citet{Sivaraman16}, we formalized the syntax and semantics of PIFO trees, developed alternate characterizations in terms of permutations on packets, established expressiveness results.
We also designed, implemented and tested embedding algorithms.
Overall, we believe our work represents the next step toward developing a programming language account of scheduling algorithms---an important topic that has mostly remained in the domain of networking and systems.
In the future, we are interested in further exploring the theory and practice of scheduling algorithms, including non-work conserving algorithms as well as applications to other domains, such as task scheduling.

\section{Data Availability Statement}

Code supporting this paper is maintained publicly on GitHub~\cite{github}.
The version submitted to the OOPSLA '23 AEC is permanently archived on Zenodo~\cite{artifact}.
A version of this paper that includes proofs of lemmas in~\S\ref{SEC:LIMITS} and~\S\ref{SEC:EMBEDDING} is on arXiv~\cite{arxiv}.

\section*{Acknowledgments}
Thanks to \'{E}va Tardos for suggesting the more general embedding algorithm for arbitrary target trees.
This work was supported by
the NSF under award 2118709, %
grant CCF-2008083, %
and grant FMiTF-1918396, the ONR under contract N68335-22-C-0411, and DARPA under contract HR001120C0107. %
T. Kappé was partially supported by the EU’s Horizon 2020 research and innovation program under Marie Skłodowska-Curie grant VERLAN (101027412). %

\ifarxiv%
\appendix
\crefalias{section}{appendix}
\section{Proofs for \texorpdfstring{\Cref{SEC:LIMITS}}{Section~\ref{SEC:LIMITS}}}%
\label{appendix:proofs-limits}

\restateflushinterleave*
\begin{proof}
By induction on $|q|$.
In the base, where $|q| = 0$, $\flush(q)$ is empty and $\mathsf{snap}(q)$ consists of empty lists; the claim holds.
For the inductive step, let $|q| > 0$, and $\pop(q) = (pkt, q')$.
By \Cref{lemma:sanity-check}, $q'$ is also well-formed.
By induction, $\flush(q')$ is an interleaving of $\mathsf{snap}(q')$.
It is also not too hard to show that to obtain $\mathsf{snap}(q)$ from $\mathsf{snap}(q')$, one must simply add $pkt$ to the front of one of the lists.
Hence, ${\flush(q) = \flush(q') \cdot pkt}$ is an interleaving of $\mathsf{snap}(q)$.
\end{proof}

\restatesimflush*
\begin{proof}
Let $R$ be a simulation such that $q_1 \mathrel{R} q_2$.
We can assume w.l.o.g.\ that if $q_1' \mathrel{R} q_2'$, then $\vdash q_1'$ and $\vdash q_2'$---otherwise, simply remove the pairs in $R$ that involve non-well-formed PIFO trees; by \Cref{lemma:sanity-check}, this new relation is still a simulation that relates $q_1$ to $q_2$.

We prove more generally that if $q_1 \mathrel{R} q_2$, then $\flush(q_1) = \flush(q_2)$, by induction on $|q_1|$.
In the base, $|q_1| = 0$, and so $\flush(q_1) = \epsilon$.
But then $\pop(q_1)$ is undefined, and therefore so is $\pop(q_2)$.
Because $q_2$ is well-formed, this can only mean, by \Cref{lemma:sanity-check}, that $|q_2| = 0$, and hence $\flush(q_2) = \epsilon$.

For the inductive step, suppose $|q_1| > 0$, and assume the claim holds for all PIFO trees of smaller size.
Let $\pop(q_1) = (pkt_1, q_1')$ and $\pop(q_2) = (pkt_2, q_2')$; these values are well-defined by \Cref{lemma:sanity-check} and the fact that $q_1$ and $q_2$ are well-formed.
Furthermore, because $R$ is a simulation, we have that $pkt_1 = pkt_2$, and $q_1' \mathrel{R} q_2'$.
It is not too hard to show that $|q_1'| + 1 = |q_1|$; hence, by induction we have that $\flush(q_1') = \flush(q_2')$.
This allows us to conclude that $\flush(q_1) = \flush(q_2)$.
\end{proof}

\restateconstructperm*
\begin{proof}
We proceed by induction on $t$.
In the base, where $t = *$, simply insert $pkt_1 = \pi(pkt_1)$ into the unique leaf of $q_\pi$ to satisfy both constraints.
For the inductive step, let $t = \mathsf{Node}(\vec{t})$ and $n = |\vec{t}|$; also, for $1 \leq i \leq n$, write $k_i = |\vec{t}[i]|$.
Because $k_1 + \cdots + k_n = n$, we can now write
\[
    pkt_1, \dots, pkt_n = pkt^1_1, \dots, pkt^1_{k_1}; pkt^2_1, \dots, pkt^2_{k_2}; \dots; pkt^n_1, \dots, pkt^n_{k_n}.
\]
Let $P_i = \{ pkt^i_1, \dots, pkt^i_{k_n} \}$, and let $\pi_i$ be the permutation on $P_i$ mediating between the order its packets appear in $pkt_{|t|} \cdots pkt_1$ and in $\pi(pkt_{|t|}) \cdots \pi(pkt_1)$---i.e., if $pkt^i_j$ is the $k$-th packet that appears in the former list, then $\pi_i$ sends it to the $k$-th packet in $P_i$ that appears in the latter.
By induction, we can then create for each $1 \leq i \leq n$ a PIFO tree $q_{\pi_i}$ such that
\begin{mathpar}
\flush(q_{\pi_i}) = \pi_i(pkt^i_{n_i}) \cdots \pi_i(pkt^i_1)
\and
\mathsf{snap}(q_{\pi_i}) = [pkt^i_1, \ldots, pkt^i_{n_i}]
\end{mathpar}
Note that for each $1 \leq i \leq n$, the relative order of the packets in $\flush(q_{\pi_i})$ already matches their order in $\flush(q_\pi)$.
We then choose $p \in \mathsf{PIFO}(\{ 1, \dots, n \})$ to simply interleave the indices in the right way---i.e., the most favorably ranked element is the index $i$ such that $\pi(pkt_1) = \pi(pkt^1_i)$, and so on.
Finally, define $\vec{q}$ by setting $\vec{q}[i] = q_{\pi_i}$, and $q_\pi = \mathsf{Internal}(\vec{q}, p)$ to satisfy the constraints.
\end{proof}

\restatepermutation*
\begin{proof}
We proceed by induction on $t$.
In the base, where $t = *$, we have that $\pi = \pi'$; it then suffices to $\push$ $\X$ to the head of the leaf PIFO, where it will be removed by the first $\pop$.

For the inductive step, let $t = \mathsf{Node}(\vec{t})$; we write $n = |t|$ and, for $1 \leq i \leq n$, we set ${k_i = |\vec{t}[i]|}$.
We know from the construction of $q_\pi$ that we can partition $P$ into $P_1, \dots, P_{n}$, and so we write ${q_\pi = \mathsf{Internal}(\vec{q}, p)}$ where $\vec{q}[i] = q_{\pi_i}$ for some permutation $\pi_i$ on $P_i$.
By the same token, $q_\pi' = \mathsf{Internal}(\vec{q}', p')$ such that $\vec{q}'[i] = q_{\pi_i'}$ for some permutation $\pi_i'$ on $P_i$.
For $1 \leq i \leq n$ we obtain by induction a sequence of $k_i$ $\push$es of $\X$ followed by $k_i$ $\pop$s of $\X$ that transform $q_{\pi_i}$ into $q_{\pi_i'}$.

For $1 \leq i \leq n$, let us write $pt^i_1, \dots, pt^i_{k_i}$ for the paths associated with these $\push$ operations.
We now build, for each $i$, a sequence of paths $pt'^i_1, \dots, pt'^1_{k_i} \in \mathsf{Path}(t)$ by setting $pt'^i_j = (i, r^i_j + r) :: pt^i_j$, where $r^i_j$ is the rank of the $j$-th $i$ in the PIFO at the head of~$q_{\pi'}$, and $r$ is sufficiently large (and independent of~$i$ and~$j$) to schedule the new index after all indices already in the PIFO\@.
If we then perform $n$ $\push$es of $\X$ as dictated by $pt'^1_1, \dots, pt'^1_{k_1}; \dots; pt'^n_1, \dots, pt'^n_{k_n}$ followed by $n$ $\pop$s, then each of the children of $q_\pi$ will have been transformed from $q_{\pi_i}$ to $q_i \preceq q_{\pi'_i}$, and the internal PIFO at the head of~$q_\pi$ will have been transformed from $p$ to $p'$ (with ranks shifted by $r$).
This means that the resulting tree $q$ is the same as $q_{\pi'}$, except some ranks may have shifted around; thus ${q \preceq q_{\pi'}}$.
\end{proof}

\restatesimulationsvsleaves*
\begin{proof}
Take a set of distinct packets $P = \{ pkt_1, \dots, pkt_{|t_1|} \}$.
We start by treating the special case where $q_1 = q_\mathsf{id}$, in which $\mathsf{id}$ is the identity permutation on $P$.
Consider any $q_2 \in \mathsf{PIFOTree}(t_2)$, and assume towards a contradiction that $q_1 \preceq q_2$.
Because the packets are distinct and $\flush(q_1) = \flush(q_2)$ by \Cref{lemma:sim-flush}, each $pkt_i$ must occur in exactly one leaf of $q_2$.
Also, since $|t_1| > |t_2|$ there exists a leaf of $q_2$ that holds at least two packets, by the pigeonhole principle.
Call these packets $pkt_i$ and $pkt_j$; without loss of generality, $pkt_i$ occurs in its leaf before $pkt_j$.

Now, let $\pi$ be the permutation that swaps $pkt_i$ and $pkt_j$ and acts as the identity on all other packets.
By \Cref{lemma:permutation}, there exists a sequence of $\push$ and $\pop$ operations that transforms $q_1 = q_\mathsf{id}$ into $q_1' = q_\pi$; this sequence furthermore $\push$es and $\pop$s only the packet $pkt_X \not\in P$.
Because $q_1 \preceq q_2$, we can apply to $q_2$ a sequence of $|t|$ $\push$es of $\X$ followed by a sequence of $|t|$ $\pop$s to obtain some PIFO tree $q_2' \in \mathsf{PIFOTree}(t_2)$ such that $q_1' \preceq q_2'$.

Because $\push$ing and $\pop$ping $\X$ in $q_2$ never changes the relative order of the packets in its leaves, and all occurrences of $\X$ introduced by the $\push$ operations are later removed by the $\pop$ operations, we have that $\mathsf{snap}(q_2) = \mathsf{snap}(q_2')$.
But by \Cref{lemma:flush-interleave} this must mean that $pkt_i$ occurs \emph{before} $pkt_j$ in $\flush(q_2')$, whereas $pkt_i$ occurs \emph{after} $pkt_j$ in $q_1' = q_\pi$.
This is a contradiction, since, by \Cref{lemma:sim-flush}, $\flush(q_1')$ must be the same as $\flush(q_2')$.
We conclude that $q_1 \not\preceq q_2$.

Now, consider the general case: let $q_1 \in \mathsf{PIFOTree}(t_1)$ and $q_2 \in \mathsf{PIFOTree}(t_2)$, and assume towards a contradiction that $q_1 \preceq q_2$.
Because $\vdash q_1$, we can apply $|q_1|$ $\pop$s to $q_1$ to obtain $q_1'$, which has no packets enqueued at its leaves; an equal number of $\pop$ operations on $q_2$ should yield $q_2'$, such that $q_1' \preceq q_2'$.
We can now execute a series of $\push$es that transform $q_1'$ into $q_1'' = q_\mathsf{id}$ as in the special case treated above; these should be mimicked by a series of $\push$es that transform $q_2'$ into $q_2''$ such that $q_1'' \preceq q_2''$, which we know to be impossible.
We thus conclude that $q_1 \not\preceq q_2$.
\end{proof}

\section{More general simulation}%
\label{appendix:more-general-simulation}

In \Cref{SEC:LIMITS}, we showed that PIFO trees where each leaf holds exactly one packet can always be permuted to achieve a different $\flush$ output, by a sequence of $\push$es followed by a sequence of $\pop$s.
In this appendix, we consider a generalization of this statement, which tells us that this statement generalizes to PIFO trees where the leaf queues hold more than one element.

In the sequel, let $pkts_1, \dots, pkts_n$ be $n$ strings of packets such that all packets are distinct.
Let $m_i = \len{pkts_i}$.
For $A\subs\{1,2,\ldots,n\}$, we define
\begin{align*}
m(A) &= \sum_{i\in A} m_i & g(A) &= \frac{(\sum_{i\in A} m_i)!}{\prod_{i\in A} m_i!}.
\end{align*}
The number $g(A)$ is a multinomial coefficient that counts the number of possible shuffles of the strings $\set{pkts_i}{i\in A}$, where a \emph{shuffle} is a string of length $m(A)$ containing each $pkts_i$ for $i\in A$ as a (possibly non-contiguous) substring.
Suppose a PIFO tree has $n$ leaves $\ell_i$, $1\le i\le n$.
Consider the class of PIFO trees $q$ for a fixed topology $t$ with $n = |t|$ satisfying
\begin{equation}\label{eq:propB}
\mathsf{snap}(q) = [pkts_1, \dots, pkts_n] \tag{B}
\end{equation}

\begin{lemma}%
\label{lem:configurations2}
Let $t \in \Topo$ with $|t| = n$.
For each of the $M=g(\{1,2,\ldots,n\})$ shuffles $\sigma$ of $pkts_1, \dots, pkts_n$, we can construct a PIFO tree $q_\sigma \in \mathsf{PIFOTree}(t)$ satisfying~\eqref{eq:propB} s.t.\ $\flush(q_\sigma) = \sigma$.
\end{lemma}
\begin{proof}
Note that a $\flush$ from of any PIFO tree produces a shuffle of the leaf queues, by \Cref{lemma:flush-interleave}.
First we argue by induction on $t$ that all shuffles can be produced.
If $t$ is a leaf, then $g(\{1\}) = 1$ and there is exactly one shuffle, which is given by populating its PIFO with $pkt_1$.

For the inductive step, let $t = \mathsf{Node}(ts)$ with $|ts| = k$.
We write
\[
    pkts_1^1, pkts_{|ts[i]|}^1, \dots, pkts_1^k, \dots, pkts_{|ts[k]|}^k = pkts_1, \dots, pkts_n
\]
For $1 \leq i \leq k$, we can construct a shuffle $\sigma_k$ of $pkts_1^i, \dots, pkts_{|ts[i]|}^i$ such that $\sigma$ is a shuffle of $\sigma_1, \dots, \sigma_k$.
By induction, we can then construct for $1 \leq i \leq k$ a PIFO tree $q_{\sigma_i} \in \mathsf{PIFOTree}(ts[i])$ such that $\flush(q_{\sigma_i}) = \sigma_i$.
All that remains is to create a PIFO $p$ with indices that interleave these children, and set $q_\sigma = \mathsf{Internal}(qs, p)$, where $qs[i] = q_{\sigma_i}$; by construction, $\flush(q_\sigma) = \sigma$.
\end{proof}

The following result generalizes \Cref{lemma:permutation}, by telling us that each of the $q_\sigma$ can be transformed into one another, up to simulation.

\begin{lemma}%
\label{lem:permutations2}
Let $\sigma$ and $\tau$ be two shuffles of $pkts_1, \dots, pkts_n$, and let $q_\sigma$ and $q_\tau$ be the PIFO trees over some fixed $t \in \Topo$ with $|t| = n$, as created in \Cref{lem:configurations2}.
Then there is a sequence of $M$ $\push$es followed by $M$ $\pop$s that maps $q_\sigma$ to $q'$ such that $q' \preceq q_\tau$.
\end{lemma}
\begin{proof}
We first describe an intermediate PIFO tree $q''$ in which the tree contains $2M$ elements, obtained with $M$ $\push$es.
In $q''$, the $i^{\text{th}}$ leaf queue contains $2m_i$ elements $pkts_i\X^{m_i}$ in that order, where $\X$ is a new packet.
The queue of an internal node in $q''$ contains the queue of the same node in $q_\sigma$ followed by the queue of that node in $q_\tau$ (with the ranks shifted, as in \Cref{lemma:permutation}).

To produce $q''$ from $q_\sigma$, we $\push$ the $M$ $\X^{m_i}$s in the order $\tau(1),\tau(2),\ldots,\tau(M)$.
In the $i^{\text{th}}$ $\push$, we add $\X^{m_i}$ to the \emph{head} of the queue at leaf $\tau(i)$; then if $s_0, \dots, s_k$ is the sequence of nodes on the path from the root down to $\tau(i)=s_k$, we add each $s_{j+1}$ to the \emph{tail} of the queue at $s_j$, $0\le j\le k-1$.
A sequence of $M$ $\pop$s will now remove all the $\X$s at the heads of the leaf queues, leaving a PIFO tree that is the same as $q_\tau$, except with some shifted ranks.
\end{proof}

An interesting consequence of \Cref{lem:configurations2,lem:permutations2} is that if we allow arbitrary computation for the ranking (that is, the ranking mechanism is Turing complete), and dummy $\push$es and $\pop$s as restructuring operations, then any $n$-leaf PIFO tree can simulate any other $n$-leaf PIFO tree.
To simulate $q_1$ by $q_2$, after every $\push$ by $q_1$, let $q_2$ insert the same data element in the same position in the corresponding leaf queue, then restructure as necessary using dummy $\push$es and $\pop$s as in \Cref{lem:permutations2} to give a configuration with the same shuffle of leaf queues as $q_1$.

\section{Proofs for \texorpdfstring{\Cref{SEC:EMBEDDING}}{Section~\ref{SEC:EMBEDDING}}}%
\label{appendix:proofs-embedding}

\restatedecomposeembedding*
\begin{proof}
If $t_1 = *$ then $t_1/\epsilon = *$, which means that $t_2/f(\epsilon) = *$ by the second property.
But by the first property $f(\epsilon) = \epsilon$, and $t_2/\epsilon = t_2$; in total, we find that $t_2 = *$.

Otherwise, let $t_1 = \mathsf{Node}(\vec{sh}_1)$ with $|\vec{sh}_1| = n$.
By the third property of homomorphic embedding, we know that for $1 \leq i \leq n$ and $addr \in \mathsf{Addr}(\vec{sh}_1[i])$, $f(i \cdot addr)$ has $f(i)$ as a prefix, i.e., $f(i \cdot addr) = f(i) \cdot w_{i \cdot addr}$ for some $w_{addr} \in \mathbb{N}^*$.
By induction on $f(i)$, we can then show that $w_{i \cdot addr} \in \mathsf{Addr}(t_2/f(i))$; hence, we set $f_i(addr) = w_{i \cdot addr}$.
Clearly this $f_i$ satisfies $f(i \cdot addr) = f(i) \cdot f_i(addr)$; it remains to show that $f_i$ embeds $t_1/i = \vec{sh}_1[i]$ inside $t_2/f(i)$.
\begin{enumerate}
    \item
    If $addr = \epsilon$, then $f(i) = f(i \cdot addr) = f(i) \cdot f_i(addr)$, which means that $f_i(addr) = \epsilon$.
    \item
    If $\vec{sh}_1[i]/addr = *$, then $t_1/(i \cdot addr) = *$ as well, which means that $t_2/f(i \cdot addr) = *$ by the second property of homomorphic embeddings.
    An inductive argument on $f(i)$ then shows that $(t_2/f(i))/f_i(addr) = t_2/(f(i) \cdot f_i(addr)) = t_2/f(i \cdot addr) = *$.
    \item
    If $addr, addr' \in \mathsf{Addr}(t_1/i)$ such that $addr$ is a prefix of $addr'$, then $i \cdot addr$ is a prefix of $i \cdot addr'$, and hence $f(i \cdot addr) = f(i) \cdot f_i(addr)$ is a prefix of $f(i \cdot addr') = f(i) \cdot f_i(addr')$ by the third property of embeddings, which means that $f_i(addr)$ is a prefix of $f_i(addr')$.
    The converse holds by a similar argument.
    \qedhere
\end{enumerate}
\end{proof}

\restateliftedembeddingwellformed*
\begin{proof}
We proceed by induction on $t_1$.
In the base, where $t_1 = *$, the claim holds because $\widehat{f}(q) = q$.
For the inductive step, let $\mathsf{Node}(\vec{sh}_1)$ and $n = |\vec{sh}_1|$.
We can then write $q = \mathsf{Internal}(\vec{q}, p)$, where $p$ is a PIFO over $\{1, \dots, n\}$ and for $1 \leq i \leq n$ we have that $\vec{q}[i] \in \mathsf{PIFOTree}(\vec{sh}_1[i])$.

We proceed by inverse induction on the prefixes of $f(i)$ for $1 \leq i \leq n$, showing more generally:
\begin{enumerate}
    \item
    If $addr$ is such a prefix, then $\vdash \widehat{f}(q)_{addr}$.
    \item
    $|\widehat{f}(q)_{addr}|$ is equal to the sum of $|p|_i$ where $1 \leq i \leq n$ and $addr$ is a prefix of $f(i)$.
\end{enumerate}
In the inner base, where $addr = f(i)$, we have $\widehat{f}(q)_{addr} = \widehat{f_i}(\vec{q}[i])$, and so $\vdash \widehat{f}(q)_{addr}$ by the outer induction hypothesis.
Furthermore, $|\widehat{f}(q)_{addr}| = |\widehat{f_i}(\vec{q}[i])| = |p|_i$, whence the second property.

For the inner inductive step, let $addr$ be a strict prefix of some $f(i)$, and suppose the root of $t_2/addr$ has $m$ children.
By induction $\vdash \widehat{f}(q)_{addr \cdot k}$ for $1 \leq k \leq m$.
Since $\vec{q}_{f,addr}[k] = \widehat{f}(q)_{addr \cdot k}$, it remains to show the second property of well-formedness, namely that for $1 \leq k \leq m$ we have that $|p_{addr}|_k = |\widehat{f}(q)_{addr \cdot k}|$.
By induction, second property tells us that the right-hand side is equal to the sum of $|p|_\ell$ where $1 \leq \ell \leq n$ and $addr \cdot k$ is a prefix of $f(\ell)$ --- but this is precisely $|p_{addr}|_k$.

To recover the second property for $\widehat{f}(q)_{addr}$, note that $|\widehat{f}(q)_{addr}|$ is, by definition, the sum of $|\widehat{f}(q)_{addr \cdot k}|$ for $1 \leq k \leq m$.
By induction, each of these terms is equal to the sum of $|p|_i$ such that $addr \cdot k$ is a prefix of $f(i)$, and so $|p|_i$ appears as a term precisely when $addr$ is a prefix of $f(i)$.
\end{proof}

\restatepreservedefined*
\begin{proof}
We proceed as before, by induction on $t_1$.
When $t_1 = *$ the claim is trivial, because $\widehat{f}(q) = q$.
In the inductive step, let $t_1 = \mathsf{Node}(\vec{sh}_1)$ with $n = |\vec{sh}_1|$.
Furthermore, $q = \mathsf{Internal}(\vec{q}, p)$.
Since $\pop(q)$ is defined, it must be the case that $p$ is not empty; let's write $i$ for the index at the head of $p$.
We then also know that $\pop(\vec{q}[i])$ is defined.

We prove that $\pop(\widehat{f}(q)_{addr})$ is defined for all prefixes $addr$ of $f(i)$, by inverse induction on $addr$; when $addr = \epsilon$, this implies the claim.
In the inner base, where $addr = f(i)$, we have that $\pop(\widehat{f}(q)_{addr}) = \pop(\widehat{f_i}(\vec{q}[i]))$, which is defined by the outer induction hypothesis.

In the inner induction step, let $addr$ be a strict prefix of $f(i)$, and let $j$ be such that $addr \cdot j$ is still a prefix of $f(i)$.
We write $\widehat{f}(q)_{addr} = \mathsf{Internal}(\vec{q}_{addr}, p_{addr})$.
Since $i$ is at the head of $p$, and $addr$ is a prefix of $f(i)$, we know that $j$ is at the head of $p_{addr}$.
Since $\vec{q}_{addr}[j] = \widehat{f}(q)_{addr \cdot j}$ and $\pop(\widehat{f}(q)_{addr \cdot j})$ is defined by induction, the claim then follows.
\end{proof}

\restatepreservepop*
\begin{proof}
We proceed by induction on $t_1$.
In the base, where $t_1 = *$, the claim holds because $\widehat{f}(q) = q$.
For the inductive step, let $t_1 = \mathsf{Node}(\vec{sh}_1)$ and $n = |\vec{sh}_1|$.
Then $q = \mathsf{Internal}(\vec{q}, p)$, where $p$ is a PIFO over $\{1, \dots, n\}$, $\vec{q}$ is a list of $n$ PIFO trees, and $\vec{q}[i] \in \mathsf{PIFOTree}(\vec{sh}_1[i])$ for $1 \leq i \leq n$.

Since $q'$ has the same topology as $q$, we also know that $q' = \mathsf{Internal}(\vec{q'}, p')$, where $p'$ is again a PIFO over $\{1, \dots, n\}$ and $\vec{q'}$ is a list of $n$ PIFO trees such that $\vec{q'}[i] \in \mathsf{PIFOtree}(\vec{sh}_1[i])$ for $1 \leq i \leq n$.
By definition of $\pop$, we then know that $\pifopop(p) = (j, p')$, such that for $1 \leq i \leq n$ we have that $\vec{q}[i] = \vec{q'}[i]$ for $1 \leq i \leq n$ with $i \neq j$, and $\pop(\vec{q}[j]) = (pkt, \vec{q'}[j])$.

The construction of $\widehat{f}(q)$ proceeds by inverse recursion on the prefixes of $f(i)$ with $1 \leq i \leq n$; thus, our proof proceeds by inverse induction on the same structure.
We argue a generalized induction hypothesis for all $addr \in \mathsf{Addr}(t_1)$ that are prefixes of some $f(i)$, and all $1 \leq j \leq n$:
\begin{enumerate}
    \item
    If $addr$ is a prefix of $f(j)$, then $\pop(\widehat{f}(q)_{addr}) = (pkt, \widehat{f}(q')_{addr})$.
    \item
    Otherwise, if $addr$ is not a prefix of $f(j)$, then $\widehat{f}(q)_{addr} = \widehat{f}(q')_{addr}$.
\end{enumerate}
Instantiating the first property for the root address $\epsilon \in \mathsf{Path}(t_1)$ then tells us that
\[
    \pop(\widehat{f}(q))
        = \pop(\widehat{f}(q)_\epsilon)
        = (pkt, \widehat{f}(q')_\epsilon)
        = (pkt, \widehat{f}(q'))
\]
We start our inner induction with the case where $addr = f(i)$ for some $1 \leq i \leq n$.
Now, if $addr = f(i)$ is a prefix of $f(j)$, then $i$ is a prefix of $j$ by the third property of embeddings, so $i = j$.
By the outer induction hypothesis, the claim applies to $\vec{q}[i]$ and so we can derive as follows:
\[
    \pop(\widehat{f}(q)_{f(j)})
        = \pop(\widehat{f_j}(\vec{q}[j]))
        = (pkt, \widehat{f_j}(\vec{q'}[j]))
        = (pkt, \widehat{f}(q')_{f(j)})
\]
Otherwise, if $addr$ is not a prefix of $f(j)$, then $i \neq j$, and so we have
\[
    \widehat{f}(q)_{f(j)}
        = \widehat{f_j}(\vec{q}[j])
        = \widehat{f_j}(\vec{q'}[j])
        = \widehat{f}(q')_{f(j)}
\]

For the inner inductive step, let $addr$ be a strict prefix of some $f(i)$.
Let $p_{addr}$ (respectively $p_{addr}'$) and $q_{f,addr}$ (resp.\ $q'_{f,addr}$) be as in the definition of $\widehat{f}(q)_{addr}$ (resp.\ $\widehat{f}(q')_{addr}$).
Note that since $\pifopop(p) = (j, p')$, we have $\pifopop(p_{addr}) = (k, p_{addr}')$, by definition of $p_{addr}$ and $p_{addr}'$.

If $addr$ is a prefix of $f(j)$, then it must be a strict prefix, by the third property of embeddings.
This means that there exists a $k$ such that $addr \cdot k$ is still a prefix of $f(j)$.
By inner induction, we also have that $\widehat{f}(q)_{addr \cdot \ell} = \widehat{f}(q')_{addr \cdot \ell}$ when $\ell \neq k$ (i.e., $addr \cdot \ell$ is not a prefix of $f(j)$), and furthermore $\pop(\widehat{f}(q)_{addr \cdot k}) = (pkt, \widehat{f}(q')_{addr \cdot k})$.
We then conclude that $\pop(\widehat{f}(q)_{addr}) = (pkt, \widehat{f}(q')_{addr})$.

Otherwise, if $addr$ is not a prefix of $f(j)$, then neither is $addr \cdot k$ for any $k$.
Hence we know by induction that $\widehat{f}(q)_{addr \cdot k} = \widehat{f}(q')_{addr \cdot k}$.
Furthermore $p_{addr} = p_{addr}'$ as well, because $p$ and $p'$ agree on all indices $k$ such that $f(k)$ has $addr \cdot k$ as a prefix (crucially, these do not include the first element $j$).
We then conclude that $\widehat{f}(q)_{addr} = \widehat{f}(q')_{addr}$.
This completes the proof.
\end{proof}

\restatepreservepush*
\begin{proof}
As before, we proceed by induction on $t_1$.
In the base, where $t_1 = *$, we have that $\widehat{f}(q) = q$ and $\widetilde{f}(pt) = pt$, and furthermore since $\push(q, pkt, pt) \in \mathsf{PIFOTree}(t_1)$, also $\widehat{f}(\push(q, pkt, pt)) = \push(q, pkt, pt)$.
The claim then holds trivially.

For the inductive step, let $t_1 = \mathsf{Node}(\vec{sh}_1)$ with $|\vec{sh}_1| = n$.
We can then write $pt = (i, r) :: pt'$ where $pt' \in \mathsf{Path}(\vec{sh}_1[i])$.
Let's write $p$ for the PIFO in the root of $q$, and $p'$ for the PIFO in the root of $\push(q, pkt, pt)$; by definition of $\push$, we have that $\pifopush(p, i, r) = p'$.

Let $addr \in \mathsf{Addr}(t_1)$ be a prefix of some $f(j)$.
We prove by inverse induction that
\begin{enumerate}
    \item
    If $addr$ is a prefix of $f(i)$, then
    \(
        \widehat{f}(\push(q, pkt, pt))_{addr}
            = \push(\widehat{f}(q)_{addr}, pkt, \widetilde{f}(pt)_{addr}).
    \)
    \item
    Otherwise, if $addr$ is not a prefix of $f(i)$, then
    \(
        \widehat{f}(\push(q, pkt, pt))_{addr} = \widehat{f}(q)_{addr}
    \)
\end{enumerate}
When $addr = \epsilon$, the first property instantiates to the claim.
In the inner base, where $addr = f(j)$ for some $1 \leq j \leq n$, there are two cases.
First, if $addr$ is a prefix of $f(i)$ then we know (by an argument analogous to the one in the proof of \Cref{lemma:preserve-pop}) that $j = i$.
The first property then reduces to
\[
    \widehat{f_i}(\push(\vec{q}[i], pkt, pt')) = \push(\widehat{f_i}(\vec{q}[i]), pkt, \widetilde{f_i}(pt'))
\]
which is true by the outer induction hypothesis.
Otherwise, if $addr$ is not a prefix of $f(i)$, then $j \neq i$.
The second property then trivializes, as $\widehat{f}(\push(q, pkt, pt))_{addr} = \widehat{f_j}(\vec{q'}[j]) = \widehat{f_j}(\vec{q}[j])$.

For the inner inductive step, let $addr$ be a strict prefix of some $f(j)$.
Now we write
\begin{mathpar}
    \widehat{f}(\push(q, pkt, pt))_{addr} = \mathsf{Internal}(\vec{q}'_{addr}, p'_{addr})
    \and
    \widehat{f}(q)_{addr} = \mathsf{Internal}(\vec{q}_{addr}, p_{addr})
\end{mathpar}
We also note that $\vec{q}_{addr}'[k] = \widehat{f}(\push(q, pkt, pt))_{addr \cdot k}$ and $\vec{q}_{addr}[k] = \widehat{f}(q)_{addr \cdot k}$, by definition.
\begin{enumerate}
    \item
    If $addr$ is a prefix of $f(i)$ it is (by the same reasoning as in \Cref{lemma:preserve-pop}) a strict prefix of $f(i)$; this means that there exists a unique $j$ such that $addr \cdot j$ is still a prefix of $f(i)$.

    When we push $pkt$ into $\widehat{f}(q)_{addr}$ with path $\widetilde{f}(pt)_{addr} = (j, r) :: \widetilde{f}(pt)_{addr \cdot j}$, we obtain a PIFO $\pifopush(p_{addr}, j, r)$, and a vector of PIFO trees $\vec{q''}$, where $q''[j] = \push(\vec{q}[j], pkt, \widetilde{f}(pt)_{addr \cdot j})$ and $\vec{q''}[k] = \vec{q}_{addr}[k]$ when $k \neq j$.
    At this point, we need to show that $\vec{q}_{addr}' = \vec{q''}$, and $\pifopush(p_{addr}, j, r) = p_{addr}'$.
    Let's focus on the vector of children first.
    If $k = j$, then $addr \cdot k$ is a prefix of $f(i)$, and so we find by inductively applying the first property that
    \[
        \vec{q}_{addr}'[k]
            = \widehat{f}(\push(q, pkt, pt))_{addr \cdot k}
            = \push(\widehat{f}(q), pkt, \widetilde{f}(pt)_{addr \cdot k})
            = \vec{q''}[k]
    \]
    Otherwise, if $k \neq j$, then we apply the second property inductively to find that
    \[
        \vec{q}_{addr}'[k]
            = \widehat{f}(\push(q, pkt, pt))_{addr \cdot k}
            = \widehat{f}(q)_{addr \cdot k}
            = \vec{q}_{addr}[k]
            = \vec{q''}[k]
    \]
    To show that $\pifopush(p_{addr}, j, r) = p_{addr}'$, note that $p_{addr}$ (resp. $p_{addr'}$) was obtained by turning the indices $k$ in $p$ (resp.\ $p'$) such that $addr$ is a prefix of $f(k)$ into $k'$, where $addr \cdot k'$ is still a prefix of $f(k)$, and deleting them if such an index does not exist.
    Now, $p$ and $p'$ agree on their contents, except that $p'$ also has $i$ inserted with rank $r$.
    Therefore $p_{addr}$ and $p_{addr}'$ must also agree on their contents, except that $p_{addr}$ has $j$ inserted with rank $r$ --- since $addr \cdot j$ is a prefix of $f(i)$.
    It then follows that $\pifopush(p_{addr}, j, r) = p_{addr}'$.
    \item
    If $addr$ is not a prefix of $f(i)$, then neither is $addr \cdot k$ for any $k$.
    We now need to show that $\vec{q}_{addr}[k] = \vec{q}_{addr}'[k]$ for all $k$, which we can do by applying the second property inductively:
    \[
        \vec{q}_{addr}[k]
            = \widehat{f}(q)_{addr \cdot k}
            = \widehat{f}(\push(q, pkt, pt))_{addr \cdot k}
            = \vec{q}_{addr}'[k]
    \]
    As for the internal PIFOs, our argument is as before: $p_{addr} = p_{addr}'$ because $p$ and $p'$ agree on all indices except the newly inserted $i$ --- but since $addr$ is not a prefix of $f(i)$, this index is removed when creating $p_{addr}$ and $p_{addr}'$ from $p$ and $p'$.
    \qedhere
\end{enumerate}
\end{proof}

\fi%

\bibliography{bibliography}

\end{document}